\documentclass[a4paper, 11pt]{article}
\usepackage[margin=1in]{geometry}
\usepackage[round]{natbib}
\usepackage[table]{xcolor}
\usepackage{wrapfig}
\usepackage{lscape}
\usepackage{rotating}
\usepackage{epstopdf}
\usepackage{pdflscape}

\usepackage{todonotes, xcolor, color}
\usepackage{marginnote, comment}
\usepackage{graphicx, fancyhdr}
\usepackage{amsthm, amsfonts, amssymb, amsmath}
\usepackage{enumitem}
\usepackage{bbm, bm, mathrsfs, mathtools}
\usepackage{authblk}
\usepackage{subfig}
\usepackage{booktabs, dcolumn, rotating, array}
\usepackage[ruled]{algorithm2e}
\usepackage{soul}
\usepackage{dsfont}
\usepackage{amstext} 
\usepackage{tikz}
\usepackage{placeins}
\newcolumntype{d}[1]{D{.}{.}{#1}}
\newcolumntype{C}{>{$}c<{$}} 
\DeclareMathOperator*{\argmin}{arg\,min}

\newtheorem{theorem}{Theorem}[section]
\theoremstyle{definition}
\newtheorem{definition}[theorem]{Definition}

\newtheorem{lemma}[theorem]{Lemma}
\newtheorem{prop}[theorem]{Proposition} 
 
\newtheorem{example}{Example} 
\newtheorem{portfolio}[theorem]{Portfolio} 

\renewcommand{\sp}[1]{\todo[color=red!10,size=\tiny, inline]{#1}}

\newcommand{\E}{{\mathbb{E}}}
\newcommand{\R}{{\mathbb{R}}}
\newcommand{\N}{{\mathbb{N}}}
\renewcommand{\P}{{\mathbb{P}}}

\newcommand{\RC}{{\mathcal{RC}}}
\newcommand{\MR}{{\mathcal{MR}}}

\newcommand{\VaR}{{\mathrm{VaR}}}
\newcommand{\ES}{{\mathrm{ES}}}

\renewcommand{\Pr}{{\boldsymbol{P}}}
\newcommand{\p}{{\boldsymbol{p}}}
\newcommand{\B}{{\boldsymbol{B}}}

\newcommand{\vv}{{\boldsymbol{v}}}
\newcommand{\vstar}{v^*}
\newcommand{\vvstar}{\vv^*}
\newcommand{\w}{{\boldsymbol{w}}}
\newcommand{\Mu}{{\boldsymbol{\mu}}}
\newcommand{\Epsilon}{{\boldsymbol{\epsilon}}}

\newcommand{\Id}{{\mathds{1}}}

\DeclareMathOperator*{\diag}{diag}

\newcommand{\x}{{\boldsymbol{x}}}
\newcommand{\z}{{\boldsymbol{z}}}

\newcommand{\bb}{{\boldsymbol{b}}}
\newcommand{\bfr}{{\boldsymbol{r}}}

\newcommand{\sdrp}{\texttt{sdrp}}
\newcommand{\msr}{\texttt{msr}}
\newcommand{\mmv}{\texttt{mmv}}
\newcommand{\gmv}{\texttt{gmv}}
\newcommand{\ew}{\texttt{ew}}
\newcommand{\grp}[1]{\texttt{grp#1}}
\newcommand{\rpa}[1]{\texttt{rpa{#1}}}
\newcommand{\rpb}[1]{\texttt{rpb#1}}

\DeclareMathOperator{\Lp}{{\mathbb{L}^2}}

\newcommand{\tikzcircle}[2][red,fill=red]{\tikz[baseline=-1.0ex]\draw[#1,radius=#2, draw=none] (0,0) circle ;}%
\newcommand{\tikzsquare}[1][red,fill=red]{\tikz[baseline=-1.0ex]\draw[#1] (0.05,0.05) rectangle (-0.15,-0.15) ;}%
\newcommand{\tikzline}[1][red,fill=red]{\tikz[baseline=-1.0ex]\draw[#1, draw=none] (0.07,0.07) rectangle (-0.21,-0.21) ;}%

\definecolor{col1}{HTML}{A6CEE3}
\definecolor{col2}{HTML}{1F78B4}
\definecolor{col3}{HTML}{B2DF8A}
\definecolor{col4}{HTML}{33A02C}
\definecolor{col5}{HTML}{FB9A99}
\definecolor{col6}{HTML}{E31A1C}
\definecolor{col7}{HTML}{FDBF6F}
\definecolor{col8}{HTML}{FF7F00}

\definecolor{colAsset1}{HTML}{1f77b4}
\definecolor{colAsset2}{HTML}{ff7f0e}
\definecolor{colAsset3}{HTML}{2ca02c}
\definecolor{colAsset4}{HTML}{d62728}
\definecolor{colAsset5}{HTML}{9467bd}
\definecolor{colAsset6}{HTML}{8c564b}
\definecolor{colAsset7}{HTML}{e377c2}
\definecolor{colAsset8}{HTML}{7f7f7f}
\definecolor{colAsset9}{HTML}{bcbd22}
\definecolor{colAsset10}{HTML}{17becf}
\definecolor{colAsset11}{RGB}{50,50,50}

\definecolor{colTable0}{HTML}{31A354}
\definecolor{colTable1}{HTML}{A1D99B}
\definecolor{colTable2}{HTML}{E5F5E0}

\definecolor{colSolver1}{HTML}{f8766d}
\definecolor{colSolver2}{HTML}{7cae00}
\definecolor{colSolver3}{HTML}{00bfc4}
\definecolor{colSolver4}{HTML}{c77cff}

\usepackage[plainpages=false,hyperfootnotes=false]{hyperref}
\hypersetup{
  colorlinks   = true, 
  urlcolor     = blue, 
  linkcolor    = red, 
  citecolor   = blue 
}

\title{\huge{Risk Budgeting Portfolios from Simulations\footnote{We thank the participants of the 3rd Insurance Data Science Conference (held virtually) for insightful discussions.}}}
\author[1]{Bernardo Freitas Paulo da Costa\footnote{E-mail: \href{mailto:bernardofpc@im.ufrj.br}{bernardofpc@im.ufrj.br}}}
\author[2]{Silvana M. Pesenti\footnote{E-mail: \href{mailto:silvana.pesenti@utoronto.ca}{silvana.pesenti@utoronto.ca}}. }
\author[3]{Rodrigo S. Targino\footnote{E-mail: \href{mailto:rodrigo.targino@fgv.br}{rodrigo.targino@fgv.br}}}

\affil[1]{\small{Institute of Mathematics, UFRJ}}
\affil[2]{\small{Department of Statistical Sciences, University of Toronto}}
\affil[3]{\small{School of Applied Mathematics, Getulio Vargas Foundation, Brazil}}

\date{\today}

\begin{document}
	\maketitle

\begin{abstract}
Risk budgeting is a portfolio strategy where each asset contributes a prespecified amount to the aggregate risk of the portfolio.
In this work, we propose an efficient
numerical framework that uses only simulations of returns for estimating risk budgeting portfolios.
Besides a general cutting planes algorithm for determining the weights of risk budgeting portfolios for arbitrary coherent distortion risk measures,
we provide a specialised version for the Expected Shortfall,
and a tailored Stochastic Gradient Descent (SGD) algorithm,
also for the Expected Shortfall.
We compare our algorithm to standard convex optimisation solvers and illustrate different risk budgeting portfolios,
constructed using an especially designed Julia package,
on real financial data and compare it to classical portfolio strategies. 
\end{abstract}


	\section{Introduction}
    Since the seminal work of Markowitz, investors are concerned with portfolio strategies that aim to maximise reward while simultaneously accounting for its riskiness. While the reward of a portfolio strategy is commonly captured via its expected return, the risk of a portfolio may be quantified via e.g., a risk or deviation measure of the aggregate portfolio. 
    
    Another important aspect of portfolio strategies is the notion of diversification. Examples of portfolios that may be viewed as diversified are the equally weighted portfolio or constant proportion strategy, such as the $60/40$ portfolio (60\% equities, 40\% stocks). These example portfolios, while they may be diversified in terms of asset allocation, are not diversified in terms of how each asset contributes to the aggregate portfolio risk. 
    In this paper, we focus on risk based diversification strategies called risk budgeting portfolios. Risk budgeting portfolios are portfolio strategies where each asset contributes a prespecified amount to the aggregate risk of the portfolio. The portfolio where each asset contributes equally to the overall risk -- termed the risk parity portfolio -- was coined by \cite{Qian2005RiskParity} and formally introduced by \cite{Maillard2010JPM}. Today, the best estimate (by the investment company Neuberger Berman\footnote{https://www.nb.com/en/global/insights/risk-parity-and-the-fallacy-of-the-single-cause}) is that these strategies have about \$120 billion under management. The conceptual difference between the risk parity portfolio and the equally weighted portfolio is that while the latter invests equally (in percentage) in each asset, the former has portfolio weights so that the risk of the portfolio is equally allocated to each asset. Thus, risk budgeting portfolios are strategies that diversify the risk of the aggregate portfolio among its assets.  
    
Risk budgeting portfolios have only been mathematically formalised recently in \cite{Maillard2010JPM}, and since then the related literature is growing. Early works on the topic focus on the volatility (variance) of the portfolio, see \cite{Maillard2010JPM} and \cite{Roncalli2013Book} for a detailed review. Extension to a multi-period setting has been considered in \cite{Li2021EJOR} using the volatility measure. \cite{Cesarone2017JGO} study the connection between risk parity portfolios with the variance as a measure of risk and the equal variance portfolio, a portfolio where the variances of each asset are equal.
Recently, \cite{Bellini2021EJOR} studied risk parity for expectile risk measures
and~\cite{anis2022EJOR} considers risk parity with additional cardinality constraints.

Here we develop a framework based on stochastic optimisation to the problem of finding a long-only risk budgeting portfolio based on convex risk measures.
Our algorithms use nothing more than simulations from the multivariate distribution of the assets' returns.
This is desirable since, for general joint distributions of assets' returns, there is no closed-form expression for the risk measure.
Furthermore, depending on the estimators for the return distributions, there might not be a closed-form expression for the probability density functions of the returns, however, one might still be able to \emph{sample} from such distributions, e.g., using Markov Chain Monte Carlo methods.
The proposed algorithms are implemented on a companion Julia package \citep{RiskBudgeting.jl} - called \texttt{RiskBudgeting} -,
which we used to generate the risk budgeting portfolios studied in this paper. When the risk measure is chosen as the Expected Shortfall (ES) (also known as Conditional Value-at-Risk (CVaR)),
we also provide specialised algorithms to take advantage of its special structure.

Although in less generality than our proposal, other authors have also used the ES in risk budgeting constructions. \cite{Boudt2012JR} study two investment strategies, the first strategy minimises the largest ES risk contribution (see Section \ref{sec:MRC}, below, for a definition), and the second portfolio strategy imposes  bounds on the ES risk contributions. In \cite{Darolles2015book_chapter} the authors claim to construct a risk budgeting portfolio based on both the Value-at-Risk (VaR) and the ES, writing the VaR or ES risk contributions as conditional expectations (see \cite{Tasche99}) and setting them to some prespecified budget. However, no detail on the considered optimisation problem is given. More recently, \cite{jurczenko2019expected} construct risk budgeting portfolios based on the minimisation of the sum of squared differences between the ES risk contributions and the target budget. This strategy was used, for example, in \cite{bruder2012managing} but does not lead to a convex optimisation problem. Moreover, it requires an explicit expression for the risk contributions, which is in contrast to the algorithms proposed in this paper that are designed to work on simulated scenarios of returns. 
    
    Other risk parity formulations (with accompanying optimisation algorithms) were proposed in \cite{Maillard2010JPM} and \cite{bai2016least}. In \cite{Maillard2010JPM}, the authors aim at minimising the sum of all squared differences between pairs of risk contributions. On the other hand, \cite{bai2016least} propose a \textit{least squares} formulation, minimising the squared difference between risk contributions and some auxiliary parameter, with the optimisation performed with respect to both the weights and the auxiliary parameter. Other formulations are found in \cite{feng2015scrip}.

The optimisation problem solved in this paper has its roots in the logarithm formulation of \cite{Maillard2010JPM},
where the risk budgeting portfolio is studied with the standard deviation as a risk measure.
This formulation was subsequently used in \cite{chaves2012efficient}, \cite{spinu2013algorithm}, and \cite{griveau2013fast},
where efficient algorithms to compute standard deviation risk budgeting portfolios are derived.
Similarly to our contribution, \cite{mausser2018long} propose the usage of convex optimisation algorithms
to find long-only risk parity portfolios based on the ES.
Differently from our approach, these authors assume the loss distribution is discrete over the simulated scenarios,
while our general cutting planes algorithm naturally handles continuous distributions. Moreover, we give a scenario decomposition algorithm that allows using very large samples.

In this work, we propose a general framework for risk budgeting
portfolios for coherent risk measures,
for which the logarithm formulation leads to a convex optimisation problem. The problem is solved through a cutting planes algorithm
that exploits the convexity of the risk measure to separate the evaluation of the risk from the optimisation itself. This allows us to deal with very large samples. Specialising for ES, we exploit the Rockafellar-Uryasev representation of ES, i.e.,
writing ES as an infimum,
for deriving algorithms for the corresponding risk budgeting portfolios.

The remainder of the paper is structured as follows.
Section~\ref{sec:rb_portfolios} defines risk contributions and risk budgeting portfolios.
In Section~\ref{sec:algorithms}, we propose three algorithms for calculating risk budgeting portfolios for coherent risk measures from simulation. Section~\ref{sec:numerical_studies} is devoted to a comparison of our cutting plane algorithm
to standard convex optimisation solvers for three different risk measures. Section~\ref{sec:alternative-strategies} discusses alternative portfolio strategies which are compared in Section~\ref{sec:application}.

\section{Risk Budgeting Portfolios} \label{sec:rb_portfolios}

We work on a probability space $(\Omega, \mathcal{A}, \P)$ and denote by $\Lp = \Lp(\Omega, \mathcal{A}, \P)$ the space of all square-integrable random variables. For a random variable $X \in \Lp$, we denote its distribution function by $F_X(x) = \P(X \le x)$ and its (left-) quantile, also called Value-at-Risk (VaR), by $F^{-1}_X(u)= \VaR_u(X) = \inf\{y \in \R~|~ F_X(y) \ge u\}$, for $u \in [0,1]$, where we adopt the convention that $\inf \emptyset = + \infty$.

\subsection{The Portfolio and its Risk Contributions}\label{sec:MRC}
We consider a two period economy
($t_0 := 0$ and $T := 1$) consisting of $d \in \N$ assets
with prices $\p:= (p_1,\ldots, p_d)$ at time $t_0$,
and corresponding (random) prices $\Pr:= (P_1, \ldots, P_d)$ at time $ T$. The investor has an initial endowment $v_0 > 0$
and invests $v_i = w_i v_0 \ge 0$ dollars in asset $i$ at time $t_0$.
We call $v_i$ the \emph{exposure} of the portfolio to asset $i$,
and $\w := (w_1, \ldots, w_d)$ the strictly positive portfolio \emph{weights} with $\sum_{i=1}^d w_i = 1$.
Since the investment is non-negative, the investor only considers buy-and-hold strategies, i.e., the number of shares remain constant.
The investor's \emph{loss} of a portfolio with exposure $\vv := (v_1, \ldots, v_d)$ at time $t_0$ is given by
\begin{equation*}
    L(\vv):= 
    - \sum_{i = 1}^d v_i \left(\frac{P_i}{p_i}-1\right)\, ,
\end{equation*}
and its risk is assessed through a \textit{coherent risk measure} $\rho\colon \Lp \to \R$ as defined in \cite{Artzner1999MF}.

\begin{definition}	
For a coherent risk measure $\rho$ and a portfolio with exposure $\vv$,
we define the \emph{marginal risk} of asset $i$, $i \in \{1, \ldots, d\}$, by  
\begin{equation*}
    \MR_i(\vv) := \frac{\partial}{\partial v_i} \rho\left(L(\vv)\right)
\end{equation*}
and the
\emph{risk contribution} of asset $i$, $i \in \{1, \ldots, d\}$, via
\begin{equation*}
	\RC_i(\vv):= v_i \,  \MR_i(\vv) \, .
\end{equation*}
\end{definition}

For simplicity, we assume throughout the text
that the joint distribution of the losses is continuous,
so that the partial derivatives above are well-defined.
Properties of the underlying distribution
that lead to differentiability of the portfolio risk,
are shown in  \cite{hong2009simulating} (Asms. 1--3)
for the Expected Shortfall (ES),
and further generalised in \cite{Tsanakas2016RA} (Prop. 1.) and \cite{Pesenti2021Cascade} (Prop. 3.) to distortion risk measures.
We also refer the reader to \cite{Grechuk2023EJOR} who proposes
the notion of \emph{extended gradients} for convex functionals
to overcome the need of requiring continuity of the underlying distributions
by proving the existence and uniqueness of allocation for non-differentiable but convex risk measures using the extended gradient
to define risk contributions.
\begin{lemma}
\label{lemma: Euler-allocation}
For a portfolio with exposure $\vv$ and a coherent risk measure $\rho$, it holds that
\begin{equation*}
	\rho(L(\vv)) = \sum_{i = 1}^d v_i \, \MR_i(\vv) 
	= \sum_{i = 1}^d \RC_i(\vv)    \, .
\end{equation*}
\end{lemma}
Note that by homogeneity of $\rho(L(\vv))$, i.e., $\rho(L(\lambda \vv)) = \lambda \rho(L(\vv))$ for all $\lambda \ge 0$, the risk contribution $\RC_i$ is equal to the Euler allocation and the Aumann-Shapley capital allocation of asset $i$, see e.g., \cite{Tasche99}.

We focus on distortion risk measures that are convex on the space of random variables, and which are defined via the Choquet integral
\begin{align*}
\rho(L)
	&= - \int_{-\infty}^0 1 - g(1 -  F_L(x)) \,\mathrm{d}x + \int_0^{+ \infty}g(1 - F_L(x))\,\mathrm{d}x\,,
\end{align*}
whenever at least one of the two integrals is finite, and where $g \colon [0,1] \to [0,1]$ is a non-decreasing and concave function satisfying $g(0) = 0$ and $g(1) = 1$. We further assume that $g$ is absolutely continuous which gives the following representation of the distortion risk measure \citep{Dhaene2012EAJ}
\begin{equation}\label{eq:rm-gamma}
    \rho(L) = \int_0^1 \gamma(u) F_L^{-1}(u) \, du\,,
\end{equation}
where $\gamma(u) = \partial_- g(x)|_{x = 1 - u}, ~ 0 < u < 1$, and $\partial_-$ denotes the derivative from the left. A well-known example of distortion risk measures is ES at level $\alpha \in (0,1)$ with $\gamma(u) = \frac{1}{1-\alpha} \Id_{\{u > \alpha\}}$.

For a portfolio with exposure $\vv$, the risk contribution to asset $i$ of a distortion risk measure with representation \eqref{eq:rm-gamma} is \citep{Tsanakas2009IME}
\begin{equation*}
	\MR_i(\vv) = \E\left[L_i\, \gamma\left(F_{L}(L)\right)\right]\,,
\end{equation*}
where $L_i:=  - v_i \left(\frac{P_i}{p_i}-1\right)$.

We further study the \emph{Entropic Value-at-Risk} (EVaR) at level $\alpha \in [0,1)$
\begin{equation*}
    EVaR_\alpha(L) := \inf_{t > 0} \; t \log \frac{1}{1 - \alpha} \E[ e^{L/t} ].
\end{equation*}
If $\alpha = 0$, this reduces to the expected value,
whereas for $\alpha \to 1$ it converges to the essential supremum of $L$
\citep{ahmadi2012entropic}.

\begin{example}
\label{ex:Gaussian-ES}
Assume the returns $\frac{P_i - p_i}{p_i}$ are multivariate Gaussian
with mean vector $\Mu$ and covariance matrix $\Sigma$.
Then, the portfolio loss follows a Gaussian distribution
with mean $\mu_L:= -\vv^\top \Mu$ and variance $\sigma_L^2 := \vv^\top \Sigma \vv$.
Moreover,
\[
\ES_\alpha \left(L(\vv)\right)
=  \mu_L + \frac{ \phi \left(\Phi^{-1}(\alpha)\right)}{1-\alpha}\sigma_L \, \text{ \ and \ } \text{EVaR}_{\alpha}(L(\vv)) = \mu_L + \sqrt{-2\log(\alpha)}\sigma_L  ,
\]
where $\Phi$ and $\phi$ denote the distribution function and density of a standard normal random variable. The risk contribution of asset $i$ for ES and EVaR are, respectively,
\begin{align*}
	\MR^{\ES}_i(\vv) &=  \E[L_i\, | \, L(\vv) \ge \VaR_\alpha(\vv)]= -\mu_i +  \frac{ \phi \left(\Phi^{-1}(\alpha)\right)}{1-\alpha} \, \frac{(\Sigma \vv)_i}{\sigma_L} \, , \\
 	\MR^{\text{EVaR}}_i(\vv) &= -\mu_i +  \sqrt{-2\log(\alpha)} \, \frac{(\Sigma \vv)_i}{\sigma_L} \, ,
\end{align*}
with $(\Sigma \vv)_i$ being the $i^\text{th}$ component of the vector $\Sigma \vv$.
\end{example}

\subsection{The Risk Budgeting Portfolio}
Assume an investor has \emph{risk appetite} $B^\dagger$ and aims to invest in a portfolio such that each asset has a (predefined) risk contribution. That is the risk appetite is such that $B^\dagger := \sum_{i = 1}^d B_i$, where $B_i>0$ corresponds to the contribution of the $i^\text{th}$ asset to the portfolio risk.
We denote $\B := (B_1, \dots, B_d)$ and call $\B$ the \emph{risk budget}.
The following definition of an investment portfolio is inspired by \cite{Roncalli2013Book}.
While the results in this section can be found, e.g., in \cite{Maillard2010JPM}, we provide short proofs for completeness.

\begin{definition}
For a risk budget $\B = (B_1, \dots, B_d)$ and corresponding risk appetite $B^\dagger  = \sum_{i = 1}^d B_i$, the \emph{risk budgeting} (RB) portfolio is  defined by the exposures $\vv$ that satisfy
\begin{equation*}
	B_i = \RC_i(\vv)\,, \qquad \text{for all} \quad i=1, \ldots, d\,.
\end{equation*}
\end{definition}
For a coherent risk measure, it holds by Lemma \ref{lemma: Euler-allocation} that $B^\dagger = \rho(\vv)$,
that is the risk appetite is equal to the portfolio risk of the RB portfolio. Thus, an investor with risk budget $\B$ who invests in a risk budgeting portfolio,
thus not only specifies the risk contribution of each asset within the portfolio but also the overall risk of the portfolio.
An example of a risk budgeting portfolio is the \emph{risk parity portfolio},
in which every asset has equal risk contribution, that is $\B := (\frac{1}{d}, \ldots, \frac{1}{d}) B^\dagger$.

\begin{prop}
\label{prop:RB-equivalent}
For a coherent risk measure $\rho$ and a risk budget $\B$, the RB portfolio is the portfolio with exposure $\vv$ that satisfies
\begin{align}
	\label{eq:RB-cross-condition}
	B_i \, \RC_j(\vv) & = B_j\, \RC_i(\vv)\,,
	\qquad \text{for all} \quad i, j =1, \ldots, d\,, \\
	\label{eq:RB-total-condition}
	B^\dagger & = \rho(L(\vv)) \,,
\end{align}
where  $B^\dagger$ is the risk appetite corresponding to the risk budget $\B$.
\end{prop}
	
\begin{proof}
Let $\vv$ be the exposure of a RB portfolio.
Then, $\vv$ fulfils $1 = \frac{B_i}{\RC_i(\vv)}$ for all $i = 1, \ldots, d$, and thus \eqref{eq:RB-cross-condition}. $B^\dagger = \rho(\vv)$ holds by Euler's Theorem.
Conversely, let $\vv$ satisfying \eqref{eq:RB-cross-condition} and \eqref{eq:RB-total-condition}, then, 
summing Equation \eqref{eq:RB-cross-condition} over $j$, and applying Euler's Theorem,
we have for all $i$
\begin{align*}
    B_i \, \rho(L(\vv))
	=\RC_i(\vv)B^\dagger
	=\RC_i(\vv)\,\rho(L(\vv))\,
\end{align*}
and $\vv$ is a RB portfolio.
\end{proof}
	
The formulation of a RB portfolio as an optimisation problem, which we present in the next proposition, was first proposed by \cite{Maillard2010JPM}.
	
\begin{theorem}
\label{prop:RB-opim}
For a coherent risk measure and a risk budget $\B$,
denote by $\vvstar$ any optimal solution  to
\begin{align}
	\label{eqn:RB-optimisation1}
	\min_{\vv} \rho(L(\vv))\, , \qquad \text{subject to} \quad & \sum_{i = 1}^d B_i \log(v_i) \ge 0\,.
\end{align}
Then $\vv: = \frac{B^\dagger}{\rho(L(\vvstar))}\, \vvstar$ is a RB portfolio with risk budget $\B$.
\end{theorem}
The constraint $\sum_{i = 1}^d B_i \log(v_i) \ge 0$ implicitly forces the $v_i$'s to be strictly positive so that the logarithm is well-defined.
This is consistent with $\log(y)$ being a concave function
taking values $-\infty$ on $y \leq 0$,
and keeps optimisation problem \eqref{eqn:RB-optimisation1} convex.

\begin{proof} 
Define the Lagrangian
\begin{equation*}
J(\vv, \lambda):= \rho(L(\vv)) - \lambda \sum_{i = 1}^d B_i \log(v_i)
\, ,
\end{equation*}
with Lagrange multiplier $\lambda \ge 0$. Then, for $i \in \{1, \ldots d\} $, taking derivatives with respect to $v_i$ and imposing the first order condition yields
\begin{align}
\label{eq:lambda}
\lambda = \frac{1}{B_i} \vstar_i \MR_i(\vvstar)
= \frac{1}{B_i} \RC_i(\vvstar)\, .
\end{align}
Hence, the optimal $\vvstar > 0$ fulfils \eqref{eq:RB-cross-condition}.
By homogeneity, we may define $\vv := \frac{B^\dagger}{\rho(L(\vvstar)}\vvstar$ so that the exposures $\vv$ satisfies $ \rho(L(\vv)) = B^\dagger$.
By Proposition~\ref{prop:RB-equivalent}, this yields a RB portfolio.

Conversely, let $v_i$  be strictly positive exposures satisfying \eqref{eq:RB-cross-condition}.
Scaling $v_i$ so that $x_i := C v_i$ satisfy $\sum_i B_i \log(x_i) = 0$, yields a vector $\x$ at the boundary of the feasible region.
By homogeneity, it holds
that all $\frac{\RC_i(\x)}{B_i}$ are equal (scaled up by $C$),
so that we can define a positive $\lambda$ via Equation~\eqref{eq:lambda}.
Given such $\lambda$, the $x_i$ are feasible,
satisfy the first order conditions,
and complementarity, and are therefore
optimal for problem~\eqref{eqn:RB-optimisation1}.
\end{proof}

Due to homogeneity, problem~\eqref{eqn:RB-optimisation1}
may have no optimal solution.
This happens if there are exposures $v_i$ such that
$\rho(L(\vv)) < 0$ and $\sum_{i=1}^d B_i \log(v_i) \geq 0$.
Then, scaling $\vv$ by a factor $M > 1$
will decrease the objective by a factor $M$,
while keeping feasibility since the left-hand side of the constraint becomes $\sum_{i=1}^d  B_i (\log(v_i) + \log (M) )= B^\dagger \log (M) + \sum_{i=1}^d  B_i \log(v_i) \ge 0$.
In particular, this is the case if there is one asset $j$ whose risk
is negative, so one can set $v_j$ arbitrarily large
and all others to one (so that their logarithm is zero).

\begin{prop}
If there exists a RB portfolio for a coherent risk measure $\rho$ and a risk budget $\B$, it is unique.
\end{prop}
\begin{proof}
Let $\x$ and $\z$ be two optimal solutions to~\eqref{eqn:RB-optimisation1}.
Since the logarithm is strictly convex, any convex combination of $\x$ and $\z$
yields a point strictly in the interior of the feasible set.
Then, the convex combination can be scaled back to the boundary by multiplying
all entries by a factor strictly smaller than one,
which yields a solution that has strictly smaller risk; 
contradicting optimality.
\end{proof}

From a practical point of view, it may be easier to provide a \textit{proportional budget}
of the portfolio loss instead of the nominal budget, i.e, $\bb = (b_1, \ldots, b_d)$ with $\sum_{i = 1}^d b_i = 1$.
For a fixed risk appetite $B^\dagger$, the corresponding proportional risk budget $\bb$ is given by $b_i = \frac{B_i}{B^\dagger}$.
In the same vein, one might be more interested in the \emph{weights} $\w$
of the assets inside the portfolio rather than their nominal values.
For this, one calculates the weights for the desired RB portfolio
by simply normalising the solution of~(\ref{eqn:RB-optimisation1})
so that it sums to one.

If an investor has an initial endowment $c v_0$,  $c > 0$,
and aims to find a RB portfolio with risk budget $c\B$, then the optimal RB weights are equal to those of a RB portfolio with endowment $v_0$ and risk budget $\B$. Therefore, in the algorithms presented in Section \ref{sec:algorithms}, we assume, without loss of generality, that $v_0=1$.

\section{Algorithms} \label{sec:algorithms}

As problem~\eqref{eqn:RB-optimisation1} is convex,
one might use convex optimisation algorithms to solve it.
Since our aim is to consider the cases when the returns have complex distributions,
in which the portfolio's risk $\rho(L(\vv))$
does not posses a simple expression as a function of $\vv$,
we present two main algorithms in this paper.
The first builds an approximation of the portfolio risk
by evaluating the returns in a (large) number of \emph{scenarios},
which is similar to a Sample Average Approximation (SAA) for risk-neutral problems.
The resulting problem is then solved via a two-step procedure,
building an approximation of the risk measure,
as a function of the investment $\vv$ using cutting planes. In the special case of ES,
one can first obtain an equivalent formulation for problem~\eqref{eqn:RB-optimisation1}
and only then apply cutting planes.
The second algorithm uses stochastic gradient descent,
by sampling scenarios and evaluating a stochastic gradient.
Both approaches are implemented in the companion Julia package \texttt{RiskBudgeting} \citep{RiskBudgeting.jl}
which is used in Sections~\ref{sec:numerical_studies} and~\ref{sec:application} to generate the risk budgeting portfolios.

After one has obtained an optimal solution $\vvstar$
to problem~\eqref{eqn:RB-optimisation1},
it can be normalised to either realise the desired total investment $v_0$,
the total risk appetite $B^\dagger$,
or to portfolio weights $\w$ that sum to one.

\subsection{Sample Average Approximation and Cutting Planes}
The idea for Sample Average Approximation (SAA) -- see \cite{prekopa1995stochastic} -- is to choose $N$ scenarios
of the realisations of the underlying random variables
and replace the expectation (or, in our case, the risk measure)
by its value in the sample.
In our problem, this amounts to sampling $N$ vectors
of (the joint distribution of) the prices $\Pr$ of the assets.
This effectively discretises the random variable $L(\vv)$
into a vector of $N$ realisations of the losses,
which we denote by $\ell^{(j)}(\vv)$ for $j = 1, \ldots, N$.
Thus, we obtain a discrete distribution
for the random variables, so one might wonder about differentiability
of the risk measure implied by the definition of the RB portfolio.
It should be noted that we use SAA to find an approximate solution
to optimisation problem~\eqref{eqn:RB-optimisation1},
and therefore we only need to approximate the risk measure.

A concrete implementation of an algorithm for solving optimisation problem~\eqref{eqn:RB-optimisation1}
depends on the chosen coherent risk measure $\rho$.
Given the scenarios of the prices and an exposure $\vv$,
the evaluation of $\rho(L(\vv))$
also yields a set of ``change-of-probabilities'' $\zeta$,
such that
$\rho(L(\vv)) = \E_\zeta[L(\vv)]$.
This shows that $\zeta$ is a subgradient of $\rho(\cdot)$ at $L(\vv)$~\citep{shapiro2009lectures},
and, using the chain rule, one may obtain the subgradient of $\rho \circ L$ at $\vv$.

\subsection{Cutting-planes for Expected Shortfall}

In the case of the ES,
one has a further possibility, since it has a compact formulation
as an optimisation problem itself.
For a given random variable $Z \in \mathbb{L}^2$,
the ES at level $\alpha \in [0, 1)$
has representation~\citep{Rockafellar-Uryasev}
\begin{equation}
  \label{eq:es_RU}
  \ES_\alpha[Z] = \min_t \  t + \frac{1}{1-\alpha} \E[(Z - t)_+]\,,
\end{equation}
where $(x)_+ = \max\{x, 0\}$ denotes the positive part. Other risk measures which admit a representation similar to \eqref{eq:es_RU}, where the risk measure can be written as an infimum include the EVaR \cite[Theorem 12]{Dentcheva2010Renyi}, which we will consider in the sequel, and higher moment risk measures \citep{Krokhmal2007higher} .

Substituting the representation of ES in problem~\eqref{eqn:RB-optimisation1},
with $L(\vv)$ in place of $Z$, we obtain an alternative optimisation problem for the risk budgeting portfolio 
\begin{subequations}
  \label{eq:rb_es_RU}
  \begin{align}
    \min_{\vv,t}  \quad &\; t + \frac{1}{1-\alpha} \E[(L(\vv) - t)_+] \\[.5ex]
     \textrm{s.t.}\quad & \sum_{i=1}^d B_i \log v_i \geq 0 .
\end{align}
\end{subequations}
Note that for fixed $\vv$ (which does not need to sum to one),
the infimum in $t$ is attained at the $\VaR_\alpha$ of the portfolio losses $L(\vv)$.

In an SAA setting and given $N$ samples for the portfolio loss, we can replace the expectation in equation \eqref{eq:rb_es_RU} by its sample average
\[
  \frac{1}{N} \sum_{j=1}^N (\ell^{(j)}(\vv) - t)_+.
\]
For a small number of scenarios $N$, and number $d$ of assets,
this discrete version of problem~\eqref{eq:rb_es_RU} can be solved directly.
However, for a relatively large number of scenarios, which is required
to handle the relevant cases when $\alpha$ is close to one,
one may again split the problem in two stages,
obtaining a decomposition amenable to cutting planes.
We keep the variables $\vv$ and~$t$ in the first stage problem, 
which becomes:
\begin{subequations}
  \label{eq:rb_es_1stage}
\begin{align}
    \min_{\vv,t}\quad   & t + Q(\vv, t) \\[.5ex]
    \textrm{s.t.}  \ & \sum_{i=1}^d B_i \log v_i \geq 0,
\end{align}
\end{subequations}
where the scenario average is replaced by the value function $Q(\vv, t)$
that will be evaluated in the second stage.
The corresponding value function $Q(\vv, t) = \frac{1}{N(1-\alpha)} \sum_{j=1}^N (\ell^{(j)}(\vv) - t)_+$
is convex, but not differentiable in $(\vv,t)$.
Therefore, problem~\eqref{eq:rb_es_1stage} can be solved
using cutting planes to approximate the function $Q$ from below.
This is done as follows:
\begin{enumerate}
  \item add an extra optimisation variable $z$ to~\eqref{eq:rb_es_1stage};
  \item add the constraint $z \geq Q(\vv, t)$;
  \item replace $Q(\vv, t)$ by $z$ in the objective function.
\end{enumerate}
The idea for using cutting planes is to approximate the constraint
$z \geq Q(\vv, t)$ by a family of linear inequalities of the form
\[
  Q(\vv, t) \geq c + \sum_{i=1}^d \pi_i v_i + \pi_t t,
\]
that are valid for $Q$, and use them for $z$ instead.
One standard way of generating such inequalities is to use the
\emph{subgradient inequalities} for convex functions. The subgradient inequalities of the value function $Q$ are
\[
  Q(\vv, t) \geq Q(\vv_0, t_0) + g^\top \left((\vv,t) - (\vv_0, t_0)\right),
\]
where $g$ is a subgradient of $Q$ at $(\vv_0, t_0)$.

In the context of our risk-parity portfolio,
we start from the optimisation problem
\begin{subequations}
\label{eq:rp_es_cp}
 \begin{align}
    \argmin_{\vv,t,z}  \quad & t + z \\[.5ex]
    \textrm{s.t.}    \ & \sum_{i=1}^n B_i \log v_i \geq 0,
\end{align}
\end{subequations}
and alternatively add constraints and optimise,
as described in Algorithm~\ref{alg:ES_cp}.
Since the point $\vv = (1, \ldots, 1); \ t = 0$
is feasible for every choice of $B_i$,
the loop can always be started from it.
\begin{algorithm}[t]
	\KwData{Budgets $B_i$ (in dollars)}
	\KwData{Tolerance $\varepsilon$}
	Set $v_i = 1$, $t = 0$  \tcp*{initial guess for allocation and VaR}
    Set $k = 0$, $z_{LB} = -\infty$ \tcp*{iteration counter \& lower bound}
    
	\While{True}
	{
		\tcc{Second stage: evaluate $Q$ and its subgradient at $(\vv, t)$:}
		$q = \frac{1}{N} \sum_{j=1}^N Q(\vv, t; \xi^{(j)})$\;
		$g = \frac{1}{N} \sum_{j=1}^N \partial Q(\vv, t; \xi^{(j)})$\;
		
		\uIf(\tcp*[f]{Found optimal solution}){$q - z_{LB} < \varepsilon$}
		{
            Normalise $\w = \vv/\textrm{sum}(\vv)$\;
            \Return{$(\w, t^*)$}
		}

        \tcc{First stage: refine approximation and re-evaluate}
		Add inequality $z \geq q + g^\top \left((\vv,t) - (\vv_0, t_0)\right)$
		to~\eqref{eq:rp_es_cp}
		
		Solve the first stage~\eqref{eq:rp_es_cp} for $(\vv, t, z)$,
		save optimal values $\vvstar$, $t^*$ and $z^*$\;

        $\vv := \vvstar$; $t := t^*$ \tcp*{Update trial point}
        $z_{LB} = z^*$ \tcp*{and lower bound}
	}
	
	\caption{Cutting Planes for Risk Budgeting portfolios under ES}
	\label{alg:ES_cp}
\end{algorithm}

At each iteration in Algorithm \ref{alg:ES_cp}, the value of the solution to the first stage optimisation, $z^*$,
corresponds to an under-approximation of $Q(\vv, t)$,
since we use only linear subgradient inequalities
and further minimise over $z$.
Thus, when evaluating $Q$ in the second stage in the next iteration,
we always obtain a larger value than $z^*$.
Finally, when the values of $z^*$ and the evaluation of $Q$ in the second stage coincide, then the first-stage problem
has evaluated $Q$ \emph{exactly} at the current point, and $z^*$ is the optimal solution; since, for all other points,
the first-stage problem underestimates $Q$.

The value function $Q(\vv,t)$ itself can be decomposed \emph{in scenarios}.
For $j \in \{1, \ldots, N\}$, denote by $\Pr^{(j)}$ a realisation of the random variable $\Pr$, and by
$\ell^{(j)}(\vv) = \xi^{{(j)}^\top} \vv$, where $\xi^{(j)} = \left(\frac{\p - \Pr^{(j)}}{\p}\right)$, a realisation of the loss.
Then the value function becomes
\begin{equation}
  \label{eq:ES_scen}
  Q(\vv, t) = \frac{1}{N} \sum_{j=1}^N \frac{({\xi^{(j)}}^\top \vv - t)_+}{1 - \alpha}
           = \frac{1}{N} \sum_{j=1}^N Q(\vv, t; \xi^{(j)})
\end{equation}
so that $Q(\vv, t; \xi^{(j)})$ can be evaluated \emph{separately} for each scenario $\xi^{(j)}$.
The same holds for the subgradients,
which for each scenario $j$ are given by
\begin{equation}
  \label{eq:ES_subgrad}
  \partial Q(\vv, t; \xi^{(j)}) =
  \begin{cases}
    (\xi^{(j)}, -1)/(1 - \alpha) & \text{if $Q(\vv, t; \xi^{(j)}) > 0$}, \\
    0                           & \text{otherwise}.
  \end{cases}
\end{equation}

\medskip

Note that this approach provides only an estimate of the portfolio risk
on the objective function.
Therefore, when using the true distribution
for evaluating the actual risk of the portfolio,
the result will differ due to statistical error.
The statistical error can be reduced by increasing the number of scenarios to evaluate $Q$ in~\eqref{eq:ES_scen}.

We add two  remarks regarding the numerical implementation.
First, cutting planes algorithms usually require a compact domain.
In our setting, since the RP inequality $\sum_i B_i \log(v_i) \geq 0$
allows $\vv$ to be unbounded, we also include (somewhat arbitrary) bounds
$v_i \leq M$, with some large $M>0$, for the portfolio composition in problem~\eqref{eq:rp_es_cp},
and check that none of these constraints is active in the returned solution.
Second, for the ES risk measure
one may include in the optimisation problem~\eqref{eq:rp_es_cp} the inequality constraint $z \geq 0$.
This constraint helps stabilise the optimisation problem
since otherwise it could diverge to very negative values of $z$ and $t$
during the initial iterations.
Further details can be checked in our implementation in~\cite{RiskBudgeting.jl}.

\subsection{Cutting planes for coherent risk measures}
\begin{algorithm}[t]
	\KwData{Budgets $B_i$ (in dollars)}
	\KwData{Tolerance $\varepsilon$}
    \KwData{Risk measure $\rho$}
	Set $\vv^0 = (1, \ldots, 1)$   \tcp*{initial guess for allocation}
    Set $k = 0$, $z^0 = -\infty$   \tcp*{iteration counter \& lower bound}
	\While{True}
	{
		\tcc{Evaluate the risk and its subgradient at $\vv^k$:}
		$q = \rho(L(\vv^k))$\;
		$g = \E_\zeta[\partial L(\vv^k)]$\;
		\uIf(\tcp*[f]{Found optimal solution}){$q - z^k < \varepsilon$}
		{
            Normalise $\vvstar = \vv^k/\textrm{sum}(\vv^k)$ \;
            \Return{$\vvstar$} \;
		}
		
		\tcc{Refine approximation and re-evaluate}
		Add inequality $z \geq q + g^\top (\vv - \vv^k)$
		to~\eqref{eqn:RB-stage1}

        Increment $k$\;
        
		Solve~\eqref{eqn:RB-stage1} for $(\vv, z)$,
		save optimal values $\vv^k$ and $z^k$
		
	}
	
	\caption{General Cutting Planes algorithm for Risk Budgeting portfolios}
	\label{alg:cp_general}
\end{algorithm}
A similar approach can be employed for a cutting-planes algorithm
that handles an arbitrary coherent risk measure $\rho$.
Since we do not assume a particular form,
such as the one we used in equation~\eqref{eq:es_RU},
we can only rewrite optimisation problem~\eqref{eqn:RB-optimisation1} as
\begin{equation}	\label{eqn:RB-stage1}
	\min_{\vv, z} \ z \,, \qquad \text{subject to} \quad \sum_{i = 1}^d B_i \log(v_i) \ge 0\,.
\end{equation}
where the extra variable $z$
will approximate the entire objective function $\rho(L(\vv))$.
Notice that, in particular, we do not obtain a scenario decomposition
as was obtained in equation~\eqref{eq:ES_scen}.

At each iteration $k$, we evaluate the objective function at $\vv^k$,
and obtain its subgradient $g^k$ at $\vv^k$.
The subgradient inequalities now translate to
\[
    z \geq \rho(L(\vv^k)) + (g^k)^\top (\vv - \vv^k)
\]
which are added to~\eqref{eqn:RB-stage1}.
By repeatedly adding more constraints,
we eventually get a good approximation for the objective function
near the optimal solution $\vv^*$, see Algorithm \ref{alg:cp_general}.

\subsection{Stochastic Gradient for Expected Shortfall}
Instead of discretising the scenarios \textit{a priori},
one could start from the two-stage decomposition in~\eqref{eq:rb_es_1stage}
but where $Q(\vv, t) = \frac{1}{1-\alpha} \E[(L(\vv) - t)_+]$
is the true expectation.
Then, it is easy to obtain a \emph{stochastic gradient} for $Q$,
given a sample of the random vector $\Pr$.
Since the objective function is $t + Q(\vv,t)$,
the actual gradient must be shifted to account for the term $t$.
In this manner, we obtain Algorithm~\ref{alg:ES_sgd},
which is a version of projected stochastic gradient descent (SGD) -- see \cite{prekopa1995stochastic}.

\begin{algorithm}[t]
	\KwData{Budgets $B_i$ (in dollars)}
	\KwData{Number of iterations $N$}
	\KwData{Steps $\alpha_k$}
	Set $v_i = 1$, $t = 0$  \tcp*{initial guess for allocation and VaR}
	Set $k = 1$             \tcp*{initial iteration number}
	\While{$k \leq N$}
	{
		Sample a scenario $\xi$ from the distribution of $P$;
		
		$g = \partial Q(\vv, t; \xi)$
		\tcp*{Subgradient of $Q$ from equation~\eqref{eq:ES_subgrad}}
		
        $(\vv,t) \mathrel{-}= \alpha_k [g + (\mathbf{0},1)]$
		\tcp*{Update current iterate}
		
		Project $\vv$ into feasible set $\sum_i B_i \log(v_i) \geq 0$.
		\tcp*{Recover feasibility}
		
        $k \mathrel{+}= 1$
		\tcp*{Update iteration count}
	}
	
	Return the average of all feasible points produced during the algorithm.
	\caption{SGD for Risk-Parity under ES}
	\label{alg:ES_sgd}
\end{algorithm}

This scenario decomposition available for the ES risk measure is due to
the Rockafellar-Uryasev formulation of the ES in Equation~\eqref{eq:es_RU}.
We remark that this decomposition is useful in the SGD setting, where we sample a different scenario at each iteration, and also in the SAA setting, in that it
allows for a large number of scenarios to be evaluated independently.
We further observe that this decomposition was possible after the introduction
of the auxiliary variable $t$ in the optimisation problem.

If one is interested in applying SGD to a different risk measure than ES,
it would be necessary
to obtain a reformulation that allows such a decomposition.
Simple examples of such risk measures are
(finite, convex) combinations of $\ES_\alpha$ at different $\alpha$ levels,
higher moment risk measures and EVaR.

\section{Numerical studies} \label{sec:numerical_studies}
In this section we compare the performance of the cutting planes algorithm described in Section \ref{sec:algorithms} to standard convex optimisation solvers on a series of risk budgeting problems, i.e. their solutions to optimisation problem \eqref{eqn:RB-optimisation1}.

\subsection{Experiment setup}
In all numerical examples in this section, we solve the risk parity portfolio, i.e.,
$B_1 = \ldots = B_d$, for different asset dimension
$d \in \{2, 5, 10, 25, 50, 100\}$,
assuming an initial endowment of $v_0 = 1$. We compare the cutting planes algorithm,
either Algorithm~\ref{alg:ES_cp} or Algorithm~\ref{alg:cp_general} depending on the risk measure considered, denoted here by CP,
with three convex optimisation solvers,
IpOpt~\citep{ipopt}, SCS~\citep{ocpb:16} and, MOSEK~\citep{mosek}.
All solvers are called from Julia:
our cutting planes algorithm uses JuMP to build the first-stage problem~\eqref{eq:rp_es_cp} or~\eqref{eqn:RB-stage1} and to add cuts;
IpOpt utilises the JuMP~\citep{JuMP2017} interface,
and SCS and MOSEK rely on Convex.jl~\citep{convexjl}.
In order to compare run times, the precision in all solvers is set to $10^{-6}$. 
All algorithms use $N \in \{ 1\,000, 2\,000, 3\,000, 4\,000, 5\,000\}$ scenarios and solve the risk parity problem 10 times.
All results are averaged over these 10 runs.

\begin{figure}[b]%
	\centering
	\includegraphics[width = 0.8\textwidth]{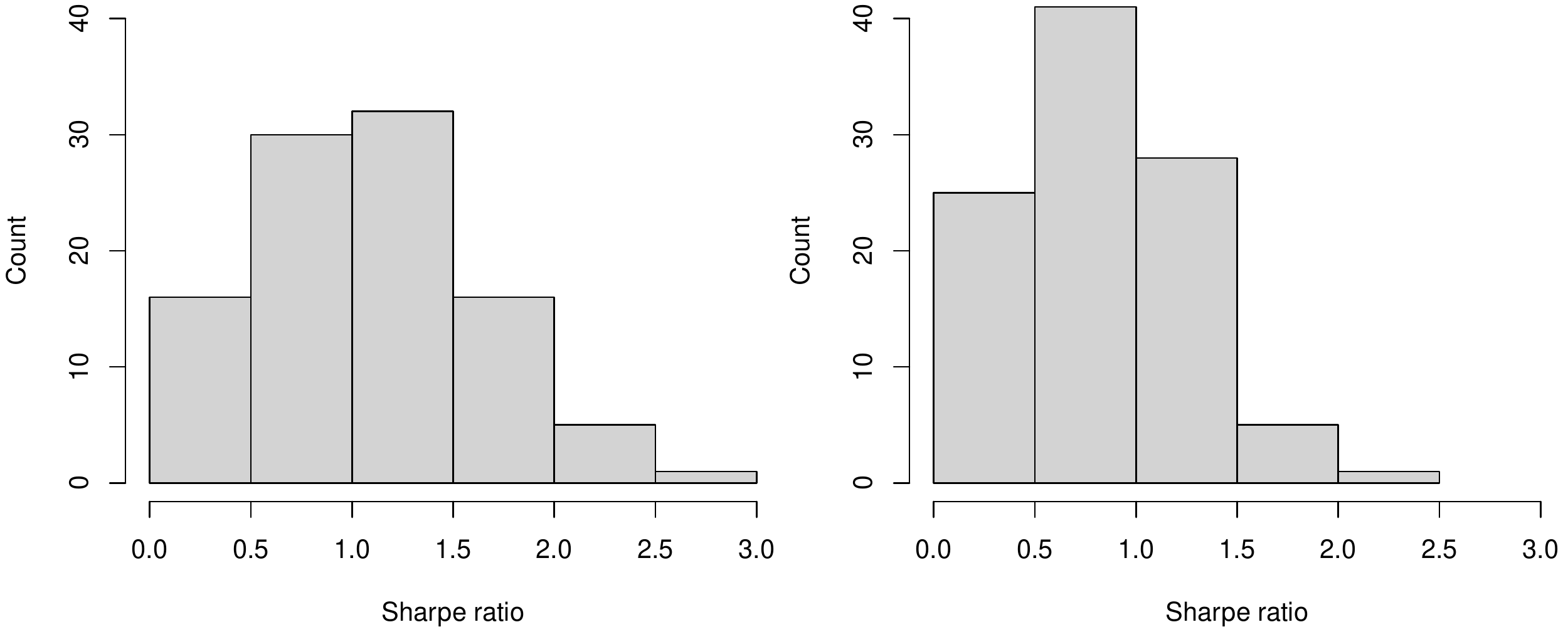}%
	\caption{Sharpe ratios for the multivariate Gaussian model (left) and the multivariate Student $t$ model (right)}%
	\label{fig:hist_sharpe}%
\end{figure}
The data generating process for the asset returns is assumed to be either multivariate Gaussian or multivariate Student $t$ with 5 degrees of freedom. This distributional assumption allows to calculate the risk contributions of the obtained portfolio weights in closed-form, thus not adding additional statistical uncertainty. For both the Gaussian and Student $t$, we generate a vector of returns $\Mu$ and a covariance matrix $\Sigma$ for the case when $d=100$ and subset its first components accordingly for the lower asset dimension cases. 
The expected returns, generated using Algorithm \ref{alg:mean_var_generation},
range from $0.7\%$ to $60\%$ and the volatilities (square root of the diagonal of $\Sigma$) vary between $10\%$ and $31\%$.
Note that even though the parameters are the same in the Gaussian and $t$ models, in the former the standard deviation for component $i$ is $\Sigma_{ii}$, while for the latter it is $\Sigma_{ii}\sqrt{\frac{\nu}{\nu-2}}$, where we set the degree of freedom to $\nu = 5$.
These values are such that the Sharpe ratio of each asset is between 0.05 and 2.7, as seen in Figure \ref{fig:hist_sharpe}.

\begin{algorithm}[t]
	\KwData{Asset dimension $d=100$}
	Set $\alpha_B$ = 1 ; 
	 $\beta_B = d^{0.4} \log(d)$ ; \\
	
	\For{$i,j = 1,\ldots, d$}
	{Sample $L_{i,j} \stackrel{iid}{\sim} Beta(\alpha_B, \beta_B)$;}
		
	Set $\Sigma = L L^T$ ; \\
	
	\For{$1,\ldots d$}
	{Sample $\displaystyle \frac{S_ i}{4} \stackrel{iid}{\sim} Beta(2,5) $\\
	Set $\Mu_i = \Sigma_{ii} \times S_i$ ;} 
	
	Return $\Mu$ and $\Sigma$
	\caption{Algorithm for generating $\Mu$ and $\Sigma$}
	\label{alg:mean_var_generation}
\end{algorithm}

\subsection{Expected Shortfall}
\label{sec:ES_rb_example}

First, we consider the performance of the cutting planes algorithm for the risk parity problem with ES and confidence levels $\alpha \in \{0.9, 0.95, 0.99 \}$.
Figures \ref{fig:cvar_dim_time_1} and \ref{fig:cvar_dim_time_4} present the log 10 of the average run time (over 10 different runs) as a function of the asset dimension $d$. The different panels correspond to increasing numbers of scenarios $N$.
All plots in Figure~\ref{fig:cvar_dim_time_1} are based on the multivariate Gaussian model with $\alpha = 0.90$, while the results in Figure~\ref{fig:cvar_dim_time_4} correspond to the multivariate Student $t$ model with $\alpha = 0.90$.
Results for the other values of $\alpha$, i.e., $\alpha = 0.95$ and $\alpha = 0.99$, are qualitatively equivalent and presented in Appendix \ref{sec:appendix_numerical_studies}. 
\begin{figure}[h]%
	\centering
	\includegraphics[width = 0.9\textwidth]{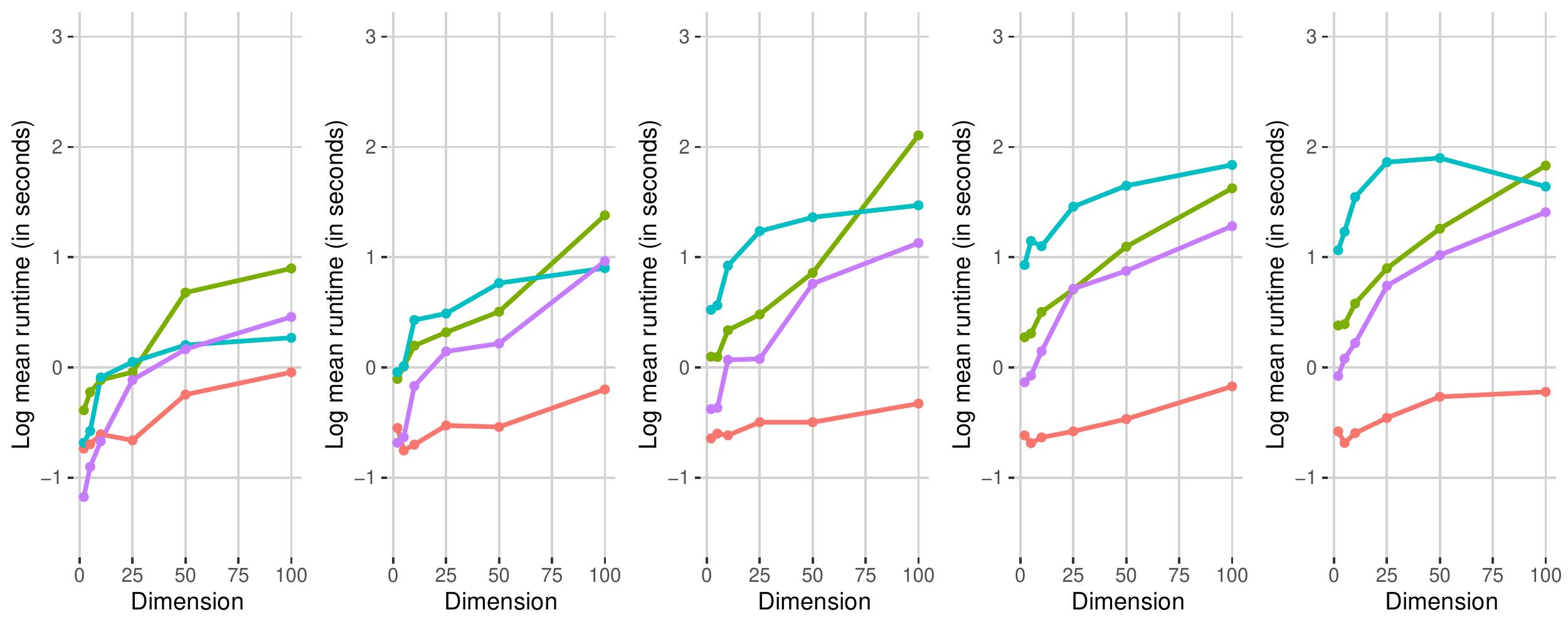}%
	\caption{Comparison of run times for different algorithms for the ES risk budgeting problem: \tikzcircle[colSolver1, fill=colSolver1]{2pt} Cutting Planes, \tikzcircle[colSolver2, fill=colSolver2]{2pt} IpOpt, \tikzcircle[colSolver3, fill=colSolver3]{2pt} MOSEK, \tikzcircle[colSolver4, fill=colSolver4]{2pt} SCS. From left to right, $N=1\,000, \ 2\,000, \ 3\,000, \ 4\,000, \ 5\,000.$ All plots are based on the multivariate Gaussian model with $\alpha = 0.90$.}%
	\label{fig:cvar_dim_time_1}%
\end{figure}
\begin{figure}[h!]%
	\centering
	\includegraphics[width = 0.9\textwidth]{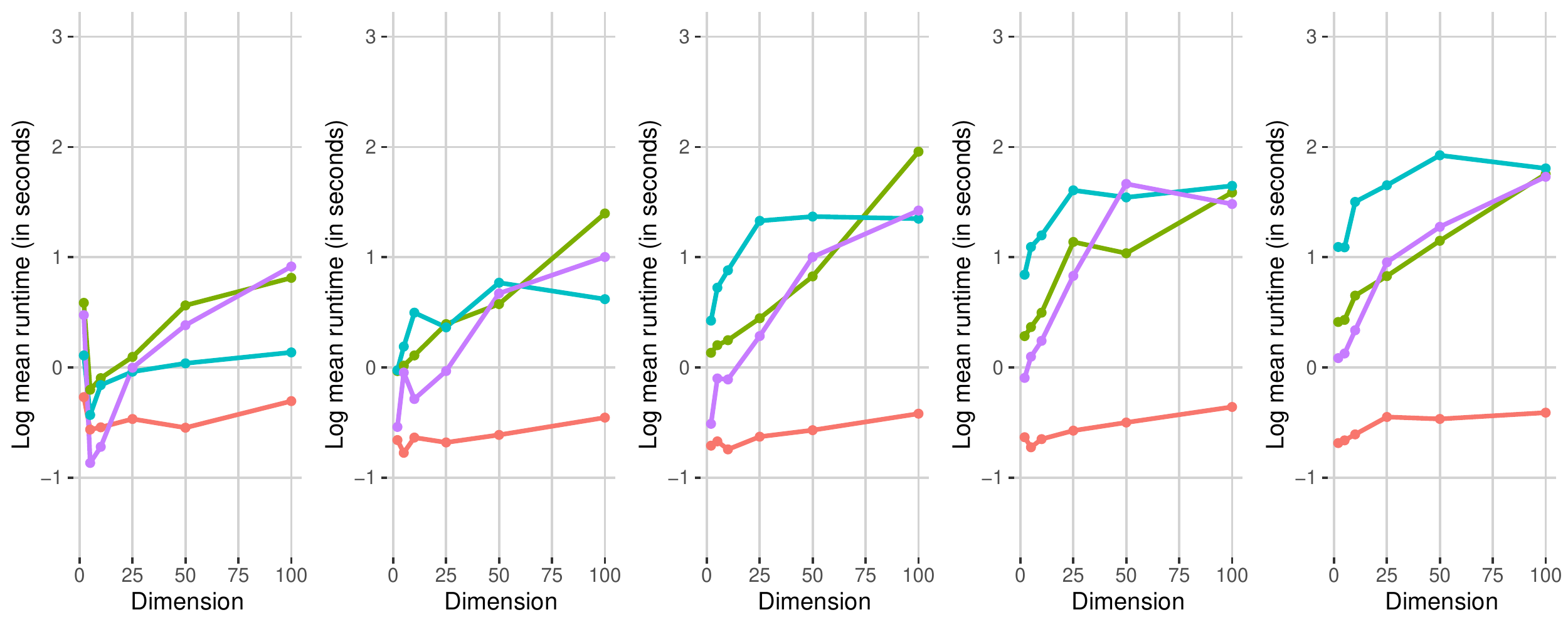}%
	\caption{Comparison of run times for different algorithms for the ES risk budgeting problem: \tikzcircle[colSolver1, fill=colSolver1]{2pt} Cutting Planes, \tikzcircle[colSolver2, fill=colSolver2]{2pt} IpOpt, \tikzcircle[colSolver3, fill=colSolver3]{2pt} MOSEK, \tikzcircle[colSolver4, fill=colSolver4]{2pt} SCS. From left to right, $N=1\,000, \ 2\,000, \ 3\,000, \ 4\,000, \ 5\,000.$ All plots are based on the multivariate Student $t$ model with $\alpha = 0.90$.}%
	\label{fig:cvar_dim_time_4}%
\end{figure}
From both run time plots, we see that the proposed cutting planes algorithm is consistently faster than all alternatives. SCS, which is faster than CP when both the asset dimension ($d$) and the number of scenarios ($N$) is small, presents run times that grow much faster than the other solvers when the asset dimension $d$ goes beyond 10. A similar behaviour is observed with IpOpt.

Differently from SCS and IpOpt, MOSEK appears to be much less sensitive to the asset dimension, exhibiting a much slower increase in run time, especially when the number of scenarios is large. However, MOSEK seams to be sensitive to increases in the number of scenarios. The proposed CP algorithm, however, is able to merge the best of all the contenders: it is both less sensitive to asset dimension and number of scenarios and moreover faster than MOSEK in all setups studied.

Overall, the simulation studies show that the proposed cutting planes algorithm is robust to the statistical model of the returns, to the number of assets in the portfolio, and to the number of scenarios, in the sense that the run time is mostly unaffected (it ranges from 0.17 seconds to 0.90 seconds in the examples of Figures \ref{fig:cvar_dim_time_1} and \ref{fig:cvar_dim_time_4}). The speed of the cutting planes algorithm allows us to study a realistic risk budgeting problem over a long time horizon in Section \ref{sec:application}.

\subsection{Entropic VaR and Distortion risk measures}
In this section we discuss the computational efficiency of the cutting planes algorithm for two different risk measures:
the Entropic Value-at-Risk (EVaR) and a distortion risk measure.
For the EVaR we use $\alpha = 0.95$,
and for the distortion risk measure we take $g(x) = \sqrt{x}$, for $x \in [0,1]$.

Table \ref{tbl:runtime_EVaR} reports the run time and the Entropic VaR risk contributions (computed as in Example \ref{ex:Gaussian-ES}) for the Gaussian distribution with $d=5$ and $N=1\,000$. We observe that SCS and IpOpt present serious problems with EVaR.
\begin{table}[ht]
\footnotesize
\centering
\begin{tabular}{lrrrrrr}
  \toprule\toprule
Solver & Run time & Asset 1 & Asset 2 & Asset 3 & Asset 4 & Asset 5\\ 
  \midrule
  CP & 1.55 & 0.109 & 0.0854 & 0.115 & 0.0667 & 0.121 \\ 
  SCS & 120.97 & 0.256 & $2.17 \times 10^{-5}$ & 0.125 & $1.36 \times 10^{-5}$ & 0.198 \\ 
  MOSEK & 1.43 & 0.256 & 0.000643 & 0.120 & 0.00176 & 0.199 \\ 
  IpOpt & 3382.49 & 0.117 & $2.03 \times 10^{-19}$ & 0.176 & $1.51 \times 10^{-19}$ & 0.298
   \\  \bottomrule\bottomrule
\end{tabular}
\caption{Run time (in seconds) and EVaR risk contributions of each one of the $d=5$ assets.}
\label{tbl:runtime_EVaR}
\end{table}
For this configuration, the IpOpt algorithm takes almost one hour and eventually stops at a solution that clearly does not satisfy the risk parity constraint, as the risk contributions are far from being the equal across assets. SCS finds a solution with similarly inaccurate risk contributions but in ``only" two minutes. The cutting planes algorithm and MOSEK present similar run times, but the CP converges to a slightly more balanced portfolio in this example, see Table \ref{tbl:runtime_EVaR}.
When the dimension of the problem is increased to $d=100$ and $N= 5\,000$, a single run of CP takes $4.76s$ while that of MOSEK takes $64.29s$.
After normalising the risk contributions to sum to one (which by homogeneity of the risk contributions is equivalent to normalising the weights), we observe that the risk contributions calculated by MOSEK are much less uniform than those by CP. This can be seen in Figure \ref{fig:entropicVaR_MOSEK_CP}, where MOSEK's solution contains one asset whose risk contributions is more than $40\%$ of the portfolio EVaR.
\begin{figure}[ht]%
	\centering
	\includegraphics[width = 0.8\textwidth]{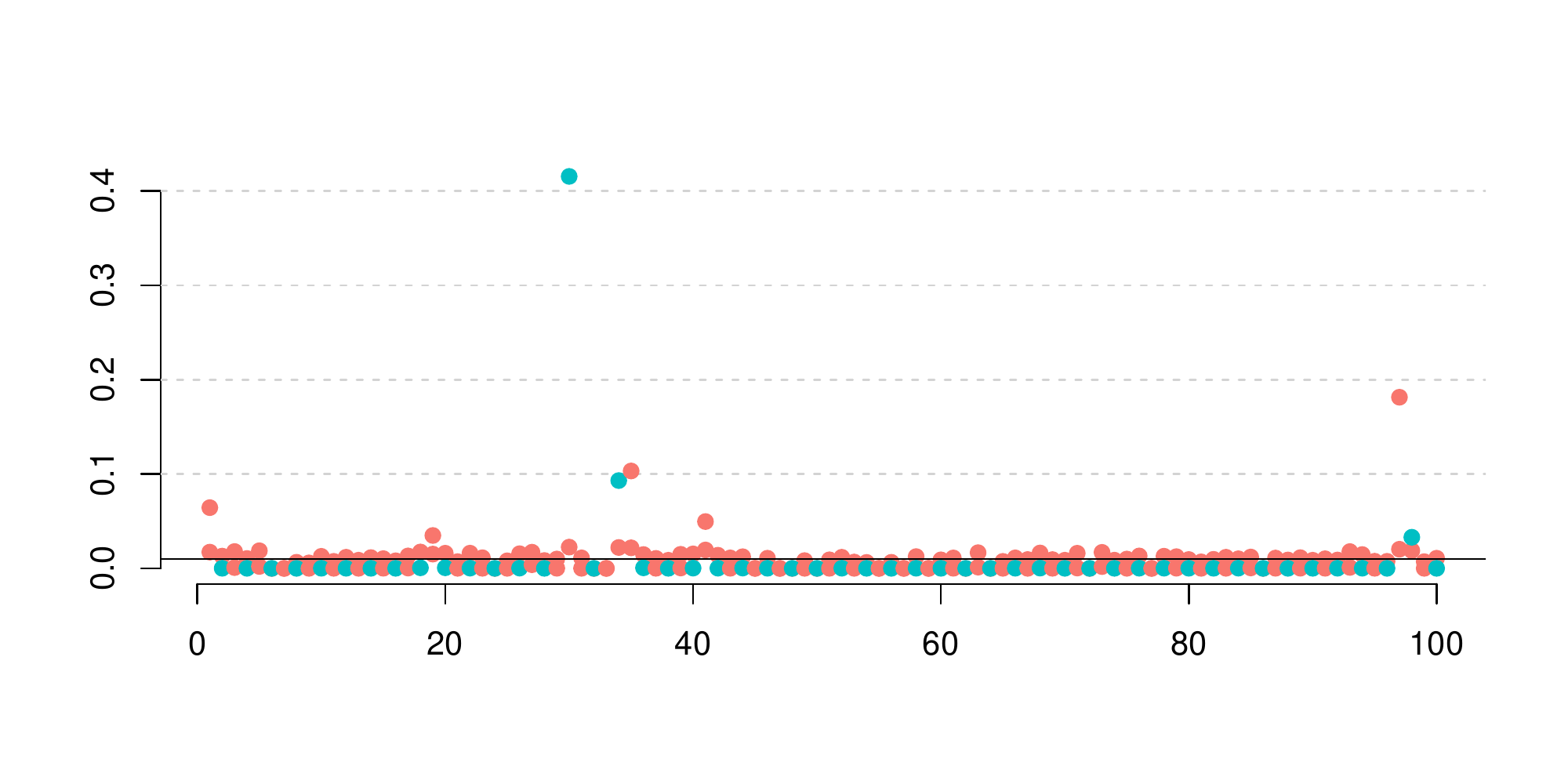}%
	\caption{Normalised risk contributions for EVaR and the Gaussian model with $d=100$ and $N=5\,000$ using \tikzcircle[colSolver1, fill=colSolver1]{2pt} Cutting Planes and \tikzcircle[colSolver3, fill=colSolver3]{2pt} MOSEK. Solid horizontal at 0.01.}%
	\label{fig:entropicVaR_MOSEK_CP}%
\end{figure}
We further show in Figure \ref{fig:entropicVaR_dim_time_FALSE} the average run times of the CP algorithm for the EVaR risk budgeting problem with multivariate Student $t$ returns. The run time results for EVaR are similar to those of ES in Section \ref{sec:ES_rb_example} in that even in the largest problem CP is able to converge (up to the required precision), on average, in less than 10 seconds. We also note that the CP algorithm is mostly unaffected by the number of scenarios $N$. 

\begin{figure}[ht]%
	\centering
	\includegraphics[width = 0.9\textwidth]{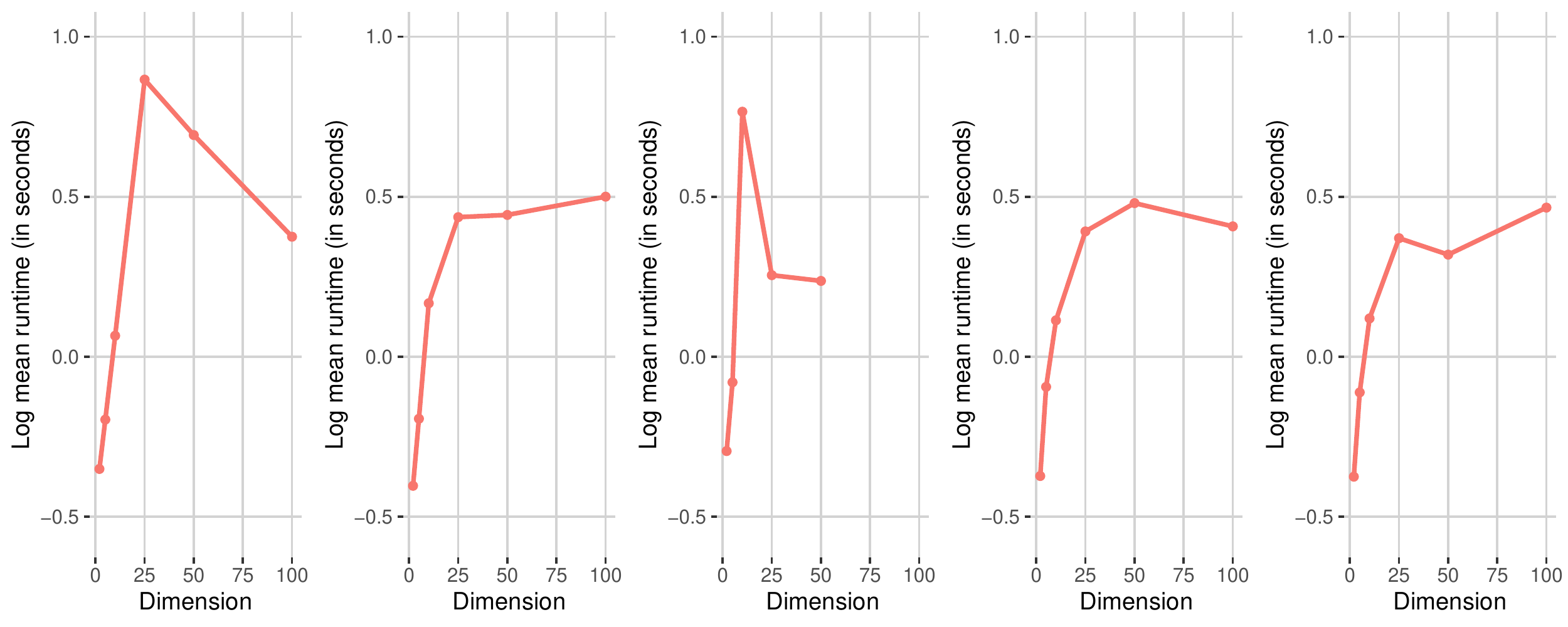}%
	\caption{Run times for the Cutting Planes algorithm for the EVaR risk budgeting problem.
    From left to right, $N=1\,000, \ 2\,000, \ 3\,000, \ 4\,000, \ 5\,000.$ All plots are based on the multivariate Student $t$ model.}%
	\label{fig:entropicVaR_dim_time_FALSE}%
\end{figure}


For the distortion risk measure with $g(x) = \sqrt{x}$, the cutting planes is the only viable solution across the proposed solvers. Even for the simplest problem of Gaussian returns, $d=2$, and $N=1\,000$, IpOpt, SCS and MOSEK run out of memory on the 32Gb computer used for the experiments. Therefore, only run times for the cutting planes algorithm are presented in Figure~\ref{fig:distortion_dim_time_TRUE}.
\begin{figure}[ht]%
	\centering
	\includegraphics[width = 0.9\textwidth]{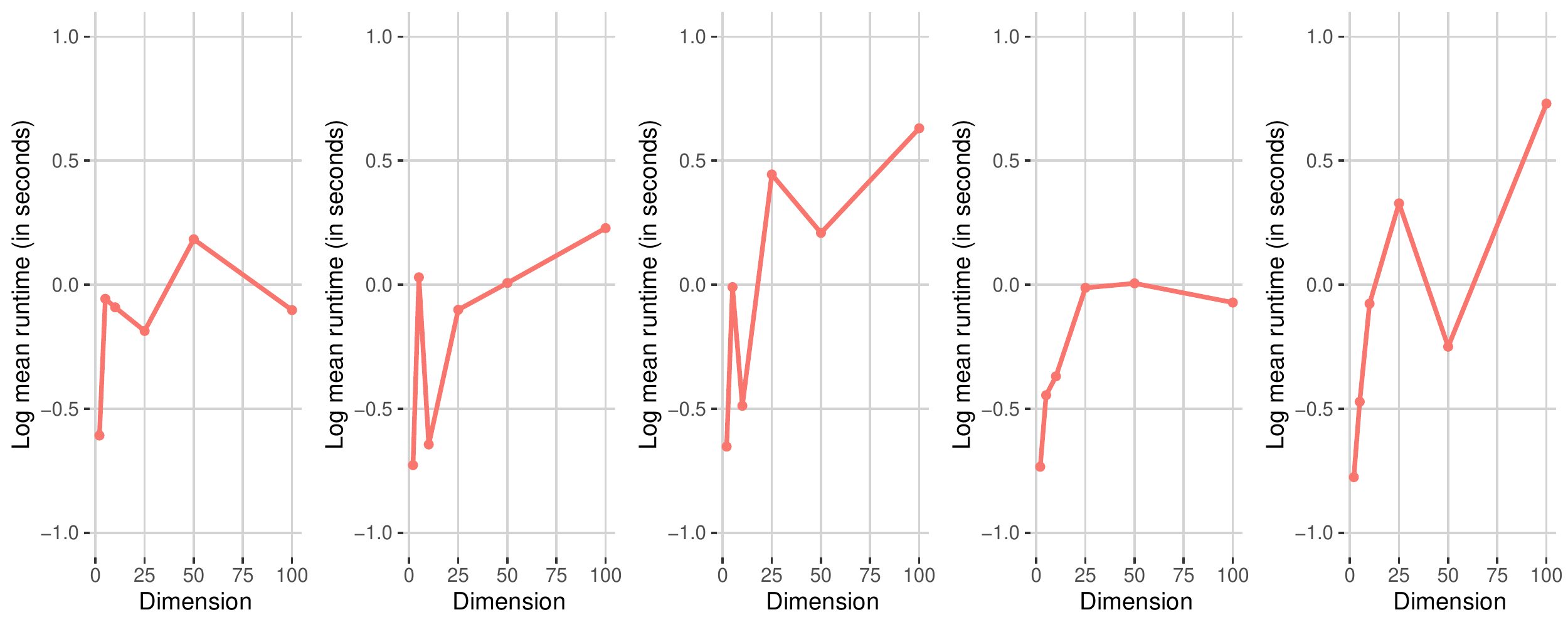}%
	\caption{Run times for the Cutting Planes algorithm for the distortion risk budgeting problem. From left to right, $N=1\,000, \ 2\,000, \ 3\,000, \ 4\,000, \ 5\,000.$ All plots are based on the multivariate Student $t$ model.}%
	\label{fig:distortion_dim_time_TRUE}%
\end{figure}
Even after averaging over 10 different runs, the results are considerably more volatile than the ones observed for ES and EVaR. This is due to a combination of factors.
First, sorting the losses, which is required for calculating distortion risk measures, may be faster or slower depending on the simulations themselves and not just their size.
Second, the number of CP iterations to achieve convergence
may be more volatile in this setting.
Nonetheless, convergence is achieved fast enough in order to enable sensitivity studies and backtests over vast time horizons and portfolio sizes.

In Appendix \ref{sec:appendix_numerical_studies}, we also report the run times for the multivariate Gaussian model, which are qualitatively similar.

\section{Investment Strategies}\label{sec:alternative-strategies}
In this section we present several portfolio selection methods and the stochastic models used to generate the RB portfolios.
For comparison purposes we set $v_0=1$ in which case $\vv = \w$ and we write $\w$ throughout this section. Further,  we consider only portfolios that are
fully invested with ($\textbf{1}^\top \w = 1$) and long-only ($\w \geq 0$).  
All portfolios are adjusted based on a set of $\Delta$ days previous to the current period, say $t$. In other words, data in the interval $[t-\Delta, \, t-1]$ is used to construct the portfolio at time $t$. We assume all portfolios are rebalanced daily. 
Daily rebalancing guarantees that the portfolios have the risk budgeting characteristic at all days. Of course, some heuristics could be used to reduce the number of trades. For example, one could trade only if the risk contributions vary by, say, more than $1\%$; when the number of shares on the previous day is kept constant. 

Apart from the proposed RB portfolios with the ES risk measure,
all  other portfolios rely only on estimates of the mean vector $\Mu$
and covariance matrix $\Sigma$ of the returns, which we estimate without imposing any assumption on its probabilistic distribution. That is, at time $t$, the mean vector and covariance matrix of the returns are estimated through their empirical counterparts
\begin{equation}\label{eq:mean-covar-estim}
    \hat \mu_{i,t} = \frac{1}{\Delta} \sum_{k=1}^{ \Delta} r_{i,t-k} \text{ \ and \ } \hat \Sigma_{i,j,t} = \frac{1}{\Delta} \sum_{k=1}^{\Delta} (r_{i,t-k} - \hat \mu_{i,t-k})(r_{j,t-k} - \hat \mu_{j,t-k}),
\end{equation}
where $\p_{t} = (p_{1,t},\ldots,p_{d,t})$ and $\bfr_{t} = (r_{1,t},\ldots,r_{d,t})$ are the vectors of prices and returns of all $d$ assets at time $t$, respectively, and $r_{i,t} = \frac{p_{i,t} - p_{i,t-1}}{p_{i,t-1}}$. Due to the small number of assets under consideration, the non-linear shrinkage estimator for the covariance matrix (see, e.g., \cite{ledoit2017numerical} and \cite{nlshrink}) did not provide any significantly different allocations, and thus are not reported.

For all portfolios introduced in the next subsection, we use the implementations presented in the vignette of the \texttt{R} package \texttt{riskParityPortfolio}  \citep{cardoso2021riskparity}.

\subsection{Mean-Variance Portfolio Selection Models}
In what follows, we give a brief description of portfolio selection models that are only
based on the mean and covariance of returns only.
For simplicity, we suppress the time index and assume portfolios are being computed for a generic time period $t$,
each period with its respective empirical mean and covariance estimates, see Equation \eqref{eq:mean-covar-estim}.

\begin{portfolio}[Maximum Sharpe ratio (\msr{})]
The maximum Sharpe ratio portfolio~\citep{sharpe1966mutual}
is implemented assuming that the risk free rate is zero.
The weights of the maximum Sharpe ratio portfolio are the solution to
\begin{align*}
	\max_{\w} \;\frac{\w^\top \Mu}{\sqrt{\w^\top \Sigma \w}} \, , \qquad \text{subject to} \quad & \w \ge \ 0 \text{ \  and \ } \textbf{1}^\top\w = 1.
\end{align*}
By homogeneity, the maximum Sharpe ratio portfolio can be transformed to a quadratic optimisation problem given by
\begin{align*}
	\min_{\vv} \vv^\top \Sigma \vv \, ,
    \qquad \text{subject to} \quad
    & \vv \ge \ 0 \text{ \  and \ } \Mu^\top \vv = 1.
\end{align*}
Then, to get weights $\w$, one simply normalises the solution $\vv$.
\end{portfolio}

\begin{portfolio}[Markowitz Mean-Variance (\mmv{})]
The weights for the \cite{markowitz1952portfolio} mean-variance portfolio are given by
\begin{align*}
	\max_{\w}\; \w^\top \Mu - \lambda \w^\top \Sigma \w \, , \qquad \text{subject to} \quad & \w \ge \ 0 \text{ \  and \ } \textbf{1}^\top\w = 1\,,
\end{align*}
where $\lambda > 0$ is user-defined risk-aversion parameter.
In our experiments we set $\lambda= \frac{1}{2}.$
\end{portfolio}

\begin{portfolio}[Global minimum variance (\gmv{})] As per its name, the weights of the global minimum variance portfolio are such that the portfolio's variance is minimised. Under assumption of no short selling and full investment restrictions, the \gmv{} weights are the solution to
\begin{align*}
	\min_{\w} \;\w^\top \Sigma \w \, , \qquad \text{subject to} \quad & \w \ge \ 0 \text{ \  and \ } \textbf{1}^\top\w = 1\,.
\end{align*}
Note that the \gmv{} is equivalent to \mmv{} under the assumption that expected returns are all equal to zero.
\end{portfolio}

\begin{portfolio}[Equal weights (\ew{})]
The equal weights portfolio assumes that the weights are equal to $w_i = 1/d$, for all $i \in \{1, \ldots, d\}$, and constant throughout time.
\end{portfolio}

\subsection{Risk Budgeting Portfolios}
We compare the portfolios introduced in the previous section with different risk budgeting portfolios. Specifically, we consider the standard deviation risk parity first studied in \cite{Maillard2010JPM} and a Gaussian risk parity with the ES risk measure, where the returns are assumed to follow a multivariate Gaussian distribution. In the next section, we describe two further risk budgeting portfolios with the ES risk measure, where the underlying returns follow dynamic conditional correlation models. For all the risk parity portfolios, the ES is computed $h = 5$ days ahead.

\begin{portfolio}[Standard deviation risk parity (\sdrp{})] Most of the literature concerning risk parity portfolios is based on the assumption that the risk is measured by the standard deviation of the portfolio, that is
\[ \rho(L(\vv)) = \sqrt{\vv^\top \Sigma \vv}\,. \]
Under this assumption, the risk budgeting optimisation problem (\ref{eqn:RB-optimisation1}) can be efficiently solved via the algorithm proposed in~\cite{feng2015scrip}. 
\end{portfolio}

\begin{portfolio}[Gaussian risk parity (\grp{})] For this portfolio we measure the risk of the portfolio via the ES for a pre-specified security level $\alpha$ and assume that the returns follow a multivariate Gaussian distribution. \end{portfolio}

Recall that for calculating risk budgeting portfolios via the algorithms in Section \ref{sec:algorithms}, we only require simulated loss scenarios at time $t +h$.
Thus, to obtain the Gaussian risk parity portfolio at time $t$, we first  estimate the mean and covariance of the returns from data at time $t$ (using data from $[t-\Delta, t-1]$ and Equation \eqref{eq:mean-covar-estim}), and then simulate Gaussian returns up to time $t+h$. 

To compare the Gaussian risk party \grp{} portfolio with risk budgeting portfolios calculated on returns that do not follow a multivariate Gaussian distribution, we consider two dynamic conditional correlation models discussed in the next section;
leading to risk-budgeting portfolios \rpa{} and \rpb{}.

\subsubsection{Dynamic Conditional Correlation Models}
In order to model the time series dynamics of the returns, we use a DCC-GARCH model (\cite{engle2002dynamic}). As seen below, the DCC-GARCH model accounts for the fact that even though returns are generally serially uncorrelated, they may present contemporaneous correlation. The heteroskedasticity in the returns is handled by univariate GARCH models (one for each asset) and the residual's correlation is dynamically modelled.

Let us denote by $R_{i,t} = \log\left(\frac{P_{i,t} }{p_{i,t-1}}\right)$ the (random) log-return of asset $i$ at time $t$. The vector with log-returns is denoted by $\mathbf{R}_{t} = (R_{1,t}, \ldots, R_{d,t})$ assumed to satisfy
\[ \mathbf{R}_{t} =  \Mu_t + \Epsilon_t, \] 
where $\Epsilon_t = \mathbf{H}_t^{1/2}\mathbf{z}_t$. The $d$-dimensional vector $\mathbf{z}_t$ is assumed to be i.i.d. and scaled, in the sense that each component of $\mathbf{z}_t$ has mean 0 and variance 1.
Therefore, $\Mu_t$ and $\mathbf{H}_t$ are, respectively, the conditional mean ($d$-dimensional vector) and covariance ($d\times d$ matrix) of $\mathbf{R}_{t}$, when data up to time $t-1$ has been observed.

The Dynamic Conditional Correlation (DCC) model proposed by \cite{engle2002dynamic} assumes that the conditional covariance of the returns can be decomposed as follows
\[ \mathbf{H_t} = \mathbf{D}_t \mathbf{S}_t \mathbf{D}_t,\]
with $\mathbf{D}_t = \diag(\sqrt{h_{11,t}}, \ldots, \sqrt{h_{dd,t}})$ being a matrix with the conditional standard deviations of the returns and $\mathbf{S}_t$ the dynamic conditional correlation matrix of dimension $d \times d$.

Positive definiteness of the time-dependent correlation matrix $\mathbf{S}_t$ is achieved through a proxy process, which, for each $d\times d$ matrix $\mathbf{Q}_t$, is defined as
\[\mathbf{Q}_{t} = (1-A-B) \bar{\mathbf{Q}} + A \mathbf{z}_{t-1} \mathbf{z}_{t-1}^\top + B \mathbf{Q}_{t-1},\]
and 
\[\bar{\mathbf{Q}} = \frac{1}{\Delta}\sum_{k=t-\Delta}^{t-1} \mathbf{z}_k \mathbf{z}_k^{\top}\]
is the unconditional matrix of the standardised errors $\mathbf{z}_t$. Recall that, at time $t$ the model is estimated using data from $t-\Delta$ to $t-1$.
To ensure stationarity and positive definiteness of $\mathbf{Q}_t$, it is assumed that the real numbers $A$ and $B$ are such that $A,B\geq 0$ and $A+B < 1$. The correlation matrix of interest is then defined as
\[\mathbf{S}_t = \diag(\mathbf{Q}_t)^{-1/2} \mathbf{Q}_t \diag(\mathbf{Q}_t)^{-1/2}, \]
where here $\diag(\mathbf{Q}_t)$ is the diagonal matrix whose entries
are the diagonal elements of $\mathbf{Q}_t$.

Different specifications for the marginal conditional variances $(h_{11, t},\ldots, h_{dd, t})$ of the returns lead to different market models. In our examples, we consider two different market models, leading to two risk budgeting portfolios. The first portfolio (\rpa{}), assumes that all conditional variances  follow a standard GARCH model, while for the second portfolio (\rpb{}), all conditional variances follow a GJR-GARCH. Both DCC models are estimated using the \texttt{R} package \texttt{rmgarch} \citep{rmgarch} through maximum likelihood.

\begin{portfolio}[DCC-GARCH (\rpa{})]
This risk budgeting portfolio consider the risk measure ES at level $\alpha$ and assumes that the conditional variances in the DCC model are given by a simple univariate GARCH(1,1) model, see, e.g., \cite{bollerslev1986generalized}. That is, for each asset $i \in \{1,\ldots,d\}$, the conditional variance of the returns at time $t$ are given by
\[ h_{ii,t} = \theta_i + a_i \epsilon_{i,t-1}^2 + b_i h_{ii,t-1}, \]
where $a_i,b_i \geq 0$, $\theta_i > 0$, and $\epsilon_{i, t}$ are i.i.d standard normally distributed.
\end{portfolio}

\begin{portfolio}[DCC-GJR-GARCH (\rpb{})] 
This risk budgeting portfolio also considers the ES risk measure at level $\alpha$.  
The difference between this portfolio (\rpb{}) and the \rpa{} is on the specification of the conditional variances in the DCC. For the \rpb{}, we assume a GJR-GARCH(1,1) model, see, e.g., \cite{glosten1993relation},
\[ h_{ii,t} = \theta_i + (a_i + c_i \mathds{1}_{\{\epsilon_{i,t-1} < 0 \}}) \epsilon_{i,t-1}^2 + b_i h_{ii,t-1}, \]
with $a_i,b_i,c_i \geq 0$, $\theta_i >0$, and $\epsilon_{i, t}$ are i.i.d standard normally distributed.. This model can be seen as an extension of the DCC-GARCH, accounting for the stylised fact that negative shocks at time $t-1$ have a stronger impact in the variance at time $t$ than positive ones.
\end{portfolio}

To calculate the optimal weights for the risk budgeting portfolios \rpa{} and \rpb{}, we use Algorithm \ref{alg:ES_cp}.

\section{Comparison of Portfolio Strategies}\label{sec:application}
In this section we analyse the long-run performance of the mean-variance portfolio models and the risk budgeting strategies under a universe of equity indices from some of the major economies in the world.
Since different investors may have access and/or preferences to different ETFs that replicate the indices, we perform all the analysis assuming the investor trades assets that perfectly replicate the indices. Henceforth, we refer to the indices as \emph{assets}. The universe of assets considered in this example is presented in Table \ref{tbl:assets}. The data has been collected from Yahoo! Finance using the \texttt{R} package \texttt{BatchGetSymbols} \cite{BatchGetSymbols}, using the tickers from the first column in Table \ref{tbl:assets}. The daily adjusted closing prices were retrieved for the period ranging from 04/Jan/2000 to 30/Dec/2020.
\begin{table}[ht]
	\centering
	\begin{tabular}{rlll}
		\toprule\toprule
		& Ticker & Name & Country \\ 
		\midrule
		1 & \verb|^|NDX & NASDAQ 100 & USA \\ 
		2 & \verb|^|GSPC & S\&P 500 &  USA \\ 
		3 & \verb|^|HSI & Hang Seng Index &  HKG \\ 
		4 & \verb|^|FTSE & FTSE 100 &  GBR \\ 
		5 & \verb|^|DJI & Dow Jones Industrial Average &  USA \\ 
		6 & \verb|^|GDAXI & DAX Performance-Index &  GER \\ 
		7 & \verb|^|RUT & Russell 2000 &  GBR \\ 
		8 & \verb|^|FCHI & CAC 40 &  FRA \\ 
		9 & \verb|^|BVSP & Ibovespa &  BRA \\ 
		10 & 000001.SS & SSE Composite Index &  CHN \\ 
		11 & \verb|^|N225 & Nikkei 225 &  JPN \\ 
		\bottomrule\bottomrule
	\end{tabular}
	\caption{Universe of assets.}
	\label{tbl:assets}
\end{table}

\begin{table}[ht]
	\centering
	\begin{tabular}{lrrrrrr}
		\toprule \toprule
		&&&&&& Max. \phantom{0}\\
		& Volatility & Return & Sharpe & $\text{VaR}_{0.05}$ &  $\text{ES}_{0.05}$ & drawdown \\ 
		\midrule
  NASDAQ 100 & 24.36\% & 18.17\% & 0.746 & -2.36\% & -3.80\% & -62.17\% \\ 
  S\&P 500 & 22.29\% & 9.11\% & 0.409 & -2.07\% & -3.62\% & -64.59\% \\ 
  Hang Seng Index & 26.15\% & 4.51\% & 0.172 & -2.52\% & -4.07\% & -73.94\% \\ 
  FTSE 100 & 21.28\% & 0.92\% & 0.043 & -2.12\% & -3.34\% & -56.73\% \\ 
  Dow Jones Industrial Average & 21.39\% & 8.65\% & 0.404 & -1.99\% & -3.45\% & -60.74\% \\ 
  DAX Performance-Index & 24.67\% & 7.21\% & 0.292 & -2.39\% & -3.89\% & -62.86\% \\ 
  Russell 2000 & 28.44\% & 8.45\% & 0.297 & -2.66\% & -4.52\% & -70.32\% \\ 
  CAC 40 & 25.18\% & 0.81\% & 0.032 & -2.48\% & -3.99\% & -69.31\% \\ 
  Ibovespa & 31.65\% & 9.69\% & 0.306 & -2.94\% & -4.75\% & -73.77\% \\ 
  SSE Composite Index & 28.48\% & 8.26\% & 0.290 & -2.96\% & -4.74\% & -80.97\% \\ 
  Nikkei 225 & 26.28\% & 4.47\% & 0.170 & -2.55\% & -4.14\% & -70.92\% \\ 
		\bottomrule \bottomrule
	\end{tabular}
\caption{Summary statistics for the assets. The volatility, return, and Sharpe ratio are annualised.}
\label{tbl:summaryStats}
\end{table}

Summary statistics for all assets across the period of study are presented in Table \ref{tbl:summaryStats}. In this table we present the annualised volatility, annualised return, annualised Sharpe ratio (ratio between return and volatility), $\text{VaR}_{0.05}$, $\text{ES}_{0.05}$, and the maximum drawdown for the entire period. 
Figure \ref{fig:all_assets} displays the values of each asset over the time period. For comparison, the assets were normalised to have initial price of 100\$.

\begin{sidewaysfigure}
\begin{center}
	\includegraphics[width=0.95\linewidth]{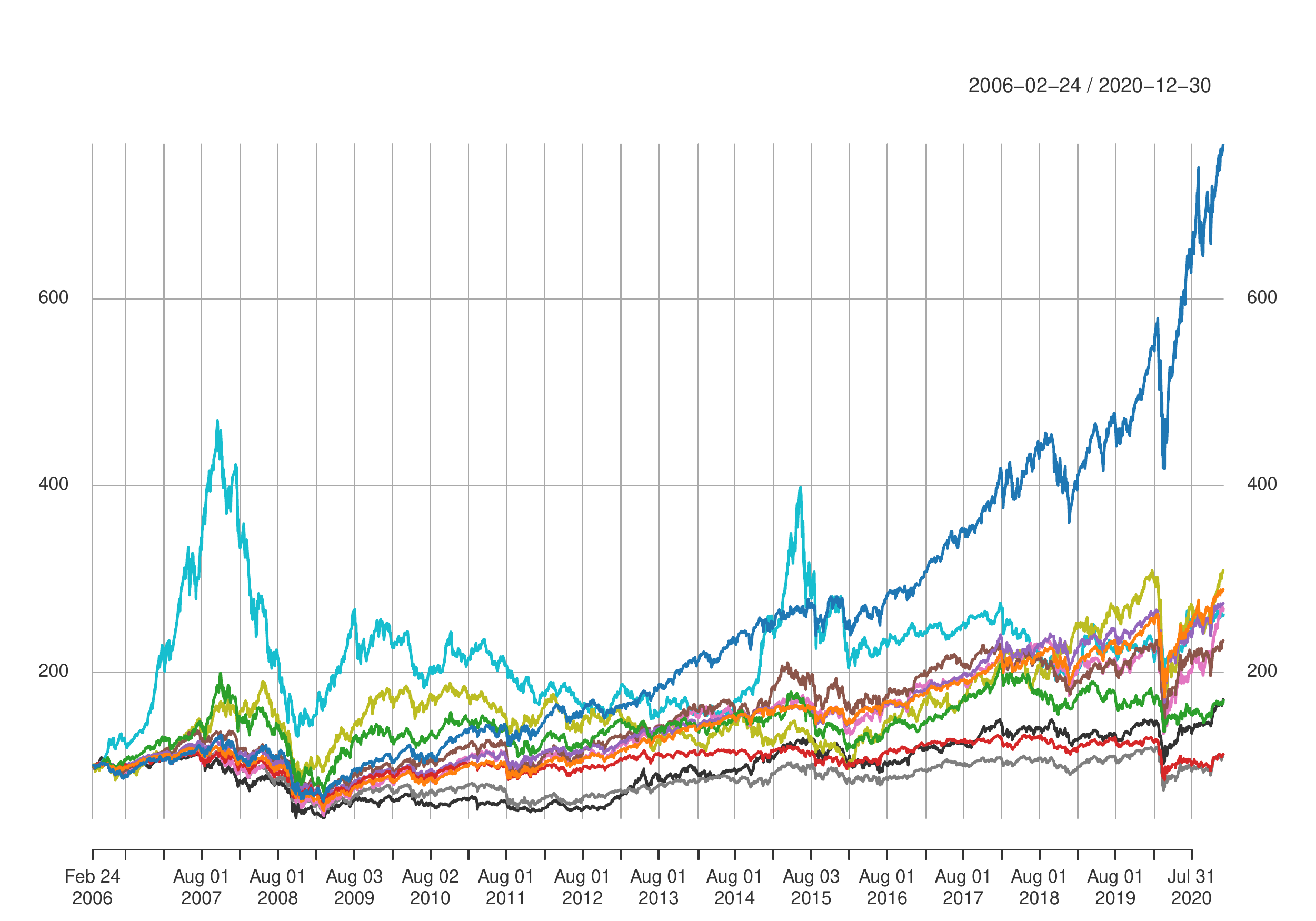}
	\caption{Asset values normalised to have initial price of 100\$. \tikzline[fill=colAsset1] NASDAQ 100, 
		\tikzline[fill=colAsset2] S\&P 500, 
		\tikzline[fill=colAsset3] Hang Seng Index, 
		\tikzline[fill=colAsset4] FTSE 100, 
		\tikzline[fill=colAsset5] Dow Jones Industrial Average, 
		\tikzline[fill=colAsset6] DAX Performance-Index, 
		\tikzline[fill=colAsset7] Russell 2000, 
		\tikzline[fill=colAsset8] CAC 40, 
		\tikzline[fill=colAsset9] Ibovespa, 
		\tikzline[fill=colAsset10] SSE Composite Index, 
		\tikzline[fill=colAsset11] Nikkei 225}
	\label{fig:all_assets}
\end{center}
\end{sidewaysfigure}

\subsection{Results}
As previously discussed, for each portfolio strategy, parameters are estimated using the previous $\Delta$ returns. For the following experiments, is set the estimation window to five ``business years'', i.e., $\Delta = 5 \times 252$. For the portfolios dependent on our algorithm (\rpa{}, \rpb{}) we assume the horizon for the risk measures is set as $h=5$ days ahead. As remarked in the opening statement of \cite{christoffersen1998horizon}: ``There is no one `magic' relevant horizon for risk management''. The choice of $h=5$ (one business week) is a compromise between the usual regulatory 10-days ahead horizon and the managerial one day ahead (see \cite{risk2020meyer}). Unless explicitly stated otherwise, we calculate the ES at significance level $\alpha = 0.95$ using $N=1\,000$ simulations at time $t+ h$. For the DCC-GARCH models, we compare two different marginal specifications, simple GARCH leading to portfolio \rpa{} and GJR-GARCH leading to portfolio \rpb{}. Both DCC-GARCH models are simulated using the function \texttt{dccsim} from the \texttt{R} package \texttt{rmgarch} \cite{rmgarch}.

For a global perspective on the performance of the portfolios under consideration, Figure \ref{fig:all_assets_wealths} presents the wealth of each asset over time, where all portfolio have an initial endowment of $100\$$. When analysing the entire time period, the Markowitz's minimum-variance (\mmv{}) portfolio shows the highest return in February 2020, a portfolio value 5 times the initial endowment. The maximum Sharpe ratio (\msr{}) portfolio returns around 4.5 times the initial endowment. All other portfolios provide returns around 3 times the initial endowment. It should be noticed, however, that the performance of portfolios are not consistent through time. For instance, in early 2013, the \msr{} portfolio shows the worst cumulative performance.

The equal weights (\ew{}) portfolio is consistently presenting the worst performance. As we see below, due to the similarity of the portfolio composition of the risk parity and the \ew{} portfolio, the risk parity portfolios present performances similar to the \ew{} portfolio.

Figure \ref{fig:all_assets_pairs} shows the daily returns for all pairs of portfolios, where each point denotes the returns for a specific day. First, we note that all pairs are positively dependent, but some, such as the portfolios based on risk parity (\sdrp{}, \grp{95}{}, \rpa{95}{}, and \rpb{95}{}) and \ew{} have very similar returns for almost all days. On the other hand, the \msr{} portfolio stands out as the one whose returns have a weaker dependence to the others. It should be noticed, that having the same daily returns is a necessary, but not sufficient, condition to guarantee two portfolios are identical. 
\begin{sidewaysfigure}
	\includegraphics[width=0.95\linewidth]{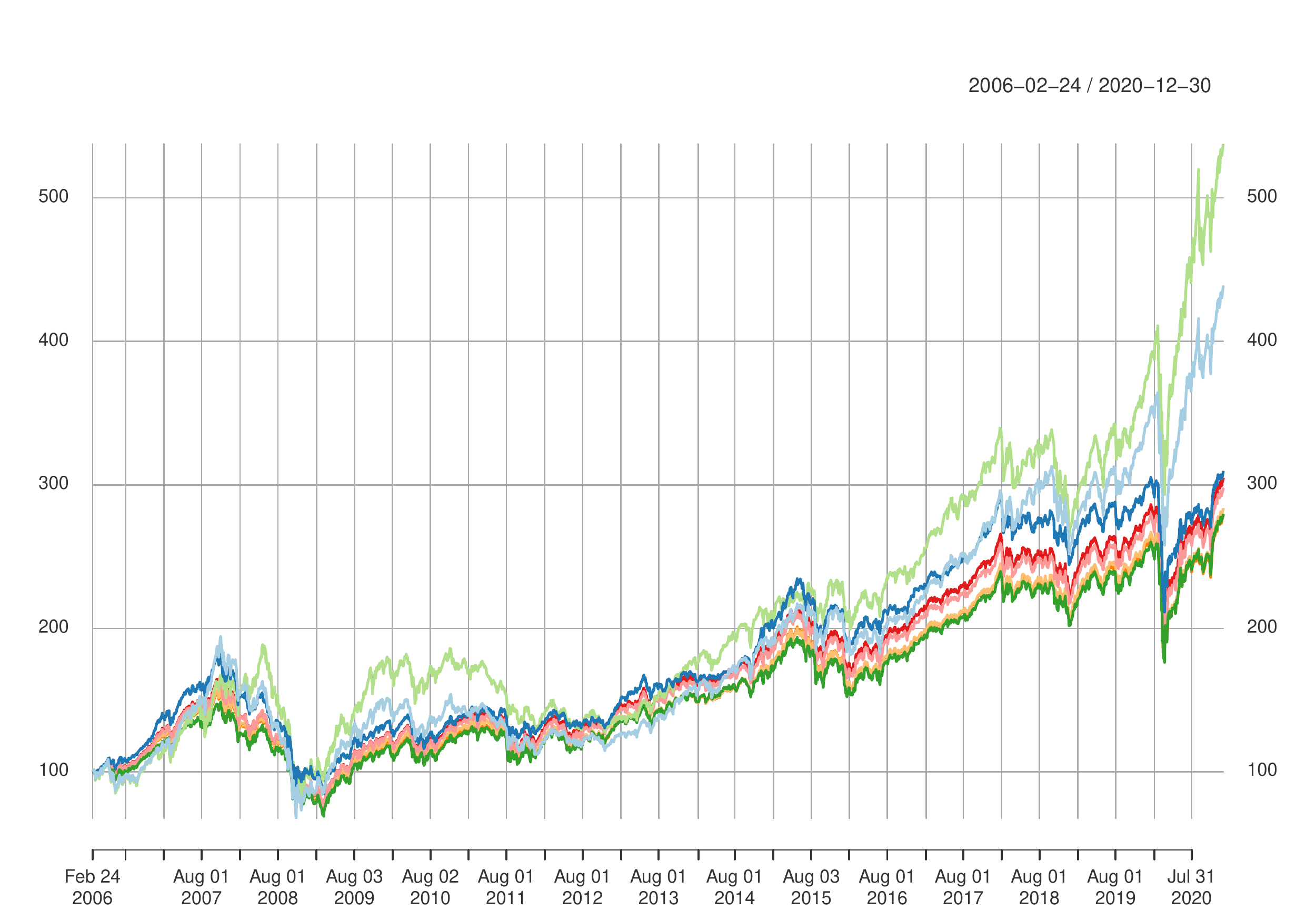}
	\caption{Values of each portfolio with initial endowment of 100\$. \tikzline[fill=col1] \msr{}, \tikzline[fill=col2] \gmv{}, 	\tikzline[fill=col3] \mmv{}, \tikzline[fill=col4] \ew{}, \tikzline[fill=col5] \sdrp{}, \tikzline[fill=col6] \grp{95}{}, \tikzline[fill=col7] \rpa{95}{}, \tikzline[fill=col8] \rpb{95}{}}
	\label{fig:all_assets_wealths}
\end{sidewaysfigure}

\begin{figure}[ht]
\begin{center}
	\includegraphics[width=0.8\linewidth]{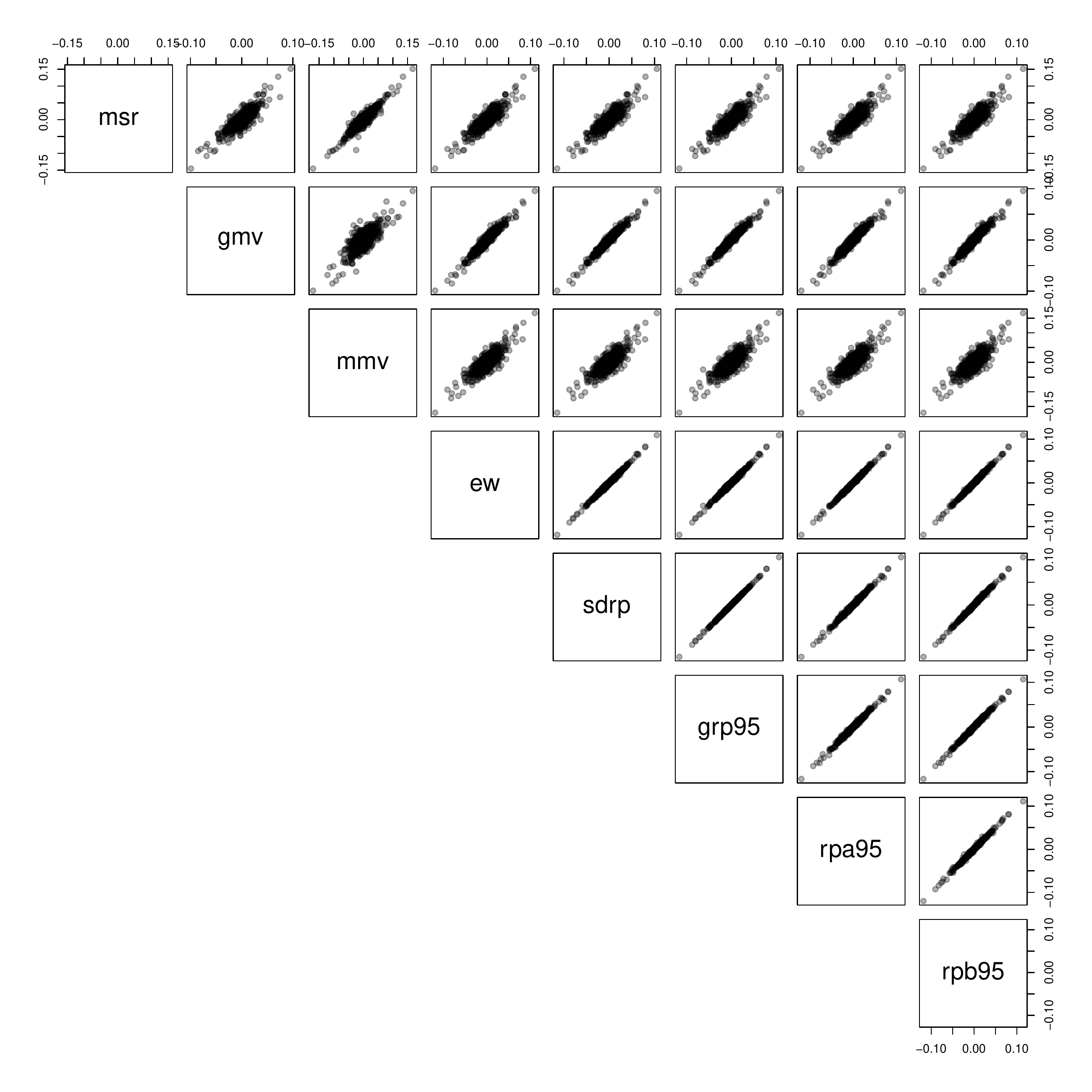}
	\caption{Scatter plot of daily returns of all pairs of portfolios.}
	\label{fig:all_assets_pairs}
\end{center}
\end{figure}

While the analysis of pairs of returns from Figure \ref{fig:all_assets_pairs} suggests that the risk parity portfolios are robust to different model assumptions, Figure \ref{fig:robustness_alpha} explores the sensitivity to the significance level $\alpha$ of the ES. In each one of the four sub-figures, we present the histogram of the returns for specific portfolios. On the top-left figure, we present the competing portfolios, from where we note a slightly heavier tail in the \msr{}. In the other sub-figures we fix the portfolio strategy (either \grp{}, \rpa{}, or \rpb{}) and vary the significance level $\alpha$. Within each sub-figure/model the histograms are almost indistinguishable, suggesting some robustness with respect to $\alpha$. 

Table \ref{tbl:portfolio_table_stats} in the Appendix \ref{sec:appendix_summary_statistics} presents the same summary statistics reported Table \ref{tbl:summaryStats} for all portfolios, including ES significance levels $\alpha \in \{0.80, 0.85, 0.90, 0.95, 0.96, 0.97, 0.98, 0.99 \}$. Although some fluctuation is seen, there does not seem to be a clear dependence of the statistics on the significance level. Over the entire 15 years of data, as we see in Figure \ref{fig:all_assets_wealths}, the \mmv{} portfolio has the highest annualised return, followed by the \msr{}. The higher returns of the \msr{} portfolio are linked to a high volatility, the highest in the sample. As expected, the global minimum variance (\gmv{}) shows the lowest volatility. Even though the \msr{} portfolio is constructed to have the maximum Sharpe ratio, since the parameters are \textit{backward-looking}, other portfolios achieve Sharpe ratios even higher, see Table \ref{tbl:stats_table}. 

Contrary to popular belief, we do not observe a stronger resistance (as measured through the Sharpe ratio) of the risk parity portfolios during moments of crisis, e.g. 2008 and 2020. This can be explained by the fact that all assets in the investment universe belong to the same asset class (Equity indices) and concurrently observed strong losses during these periods. Additionally, the investment strategies studied are bounded to be long-only and fully invested.

Among the proposed risk parity portfolios (\sdrp{}, \grp{}, \rpa{}, and \rpb{}), their annualised volatilities vary between $18.44\%$ and $19.33\%$, which are smaller than those of \mmv{} and \msr{}. Of the two DCC-GARCH risk parity models (\rpa{}{} and \rpb{}{}), the volatility adjusted performance (Sharpe ratio) is uniformly lower, across all $\alpha$ levels, for the GJR-GARCH model (\rpb{}{}). Nonetheless, both models' Sharpe are higher than the \ew{}'s and smaller than those of other portfolios. The maximum drawdown for the risk parity portfolios (\sdrp{}, \grp{}, \rpa{}, and \rpb{}) is on par with the equal weights portfolio (\ew{}) and significantly higher than those of the best performing portfolios (\msr{} and \mmv{}). A similar result is observed for the respective $\text{VaR}_{0.05}$ and $\text{ES}_{0.05}$.

\begin{figure}[ht]%
	\centering
	\subfloat[\centering{\tikzsquare[fill=col1] \msr{}, \tikzsquare[fill=col2] \gmv{}, \tikzsquare[fill=col3] \mmv{}, \tikzsquare[fill=col4] \ew{}, \tikzsquare[fill=col5] \sdrp{}}]{{\includegraphics[width=0.4\linewidth]{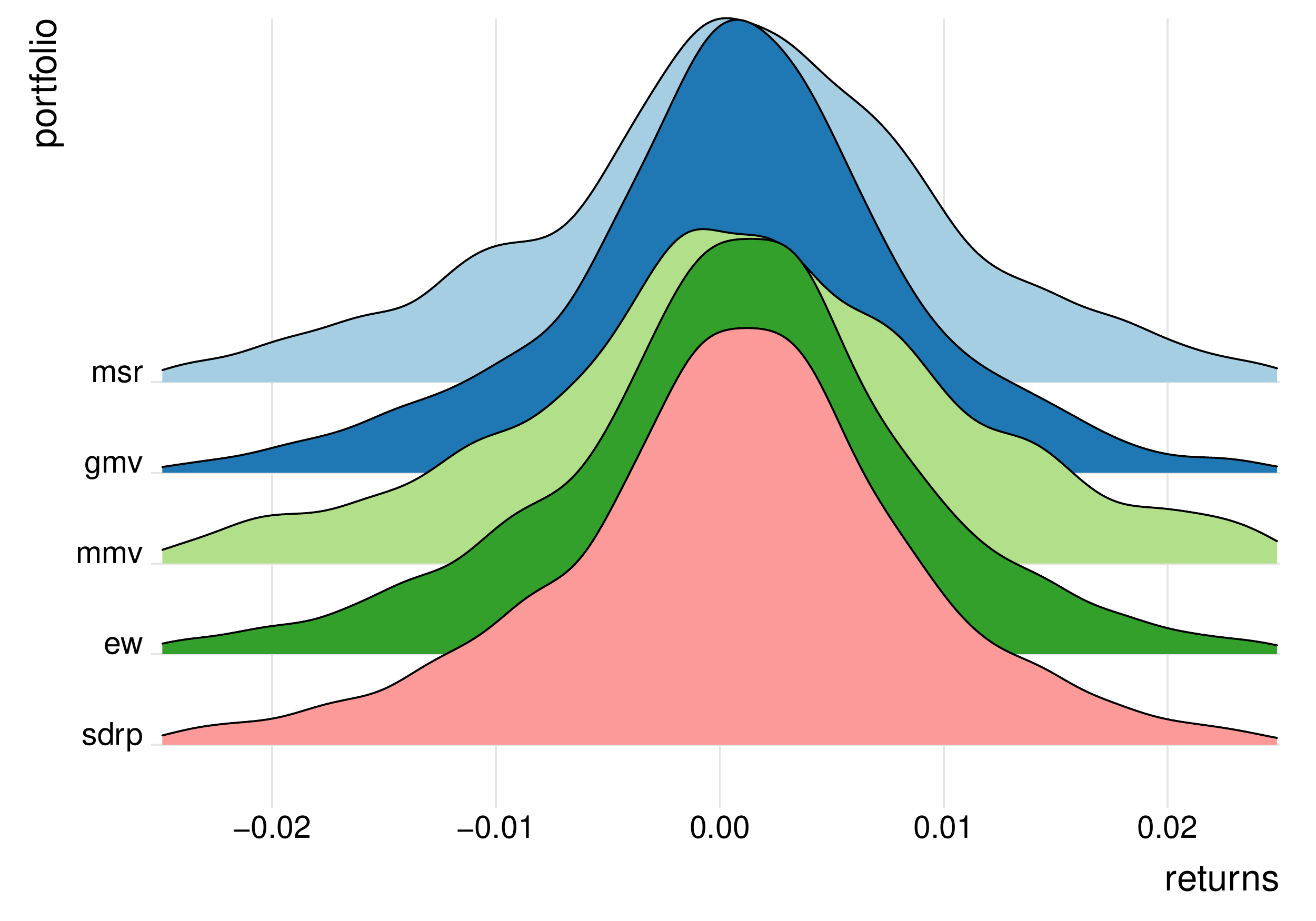} }}%
	\qquad
	\subfloat[\centering \grp{} portfolios]{{\includegraphics[width=0.4\linewidth]{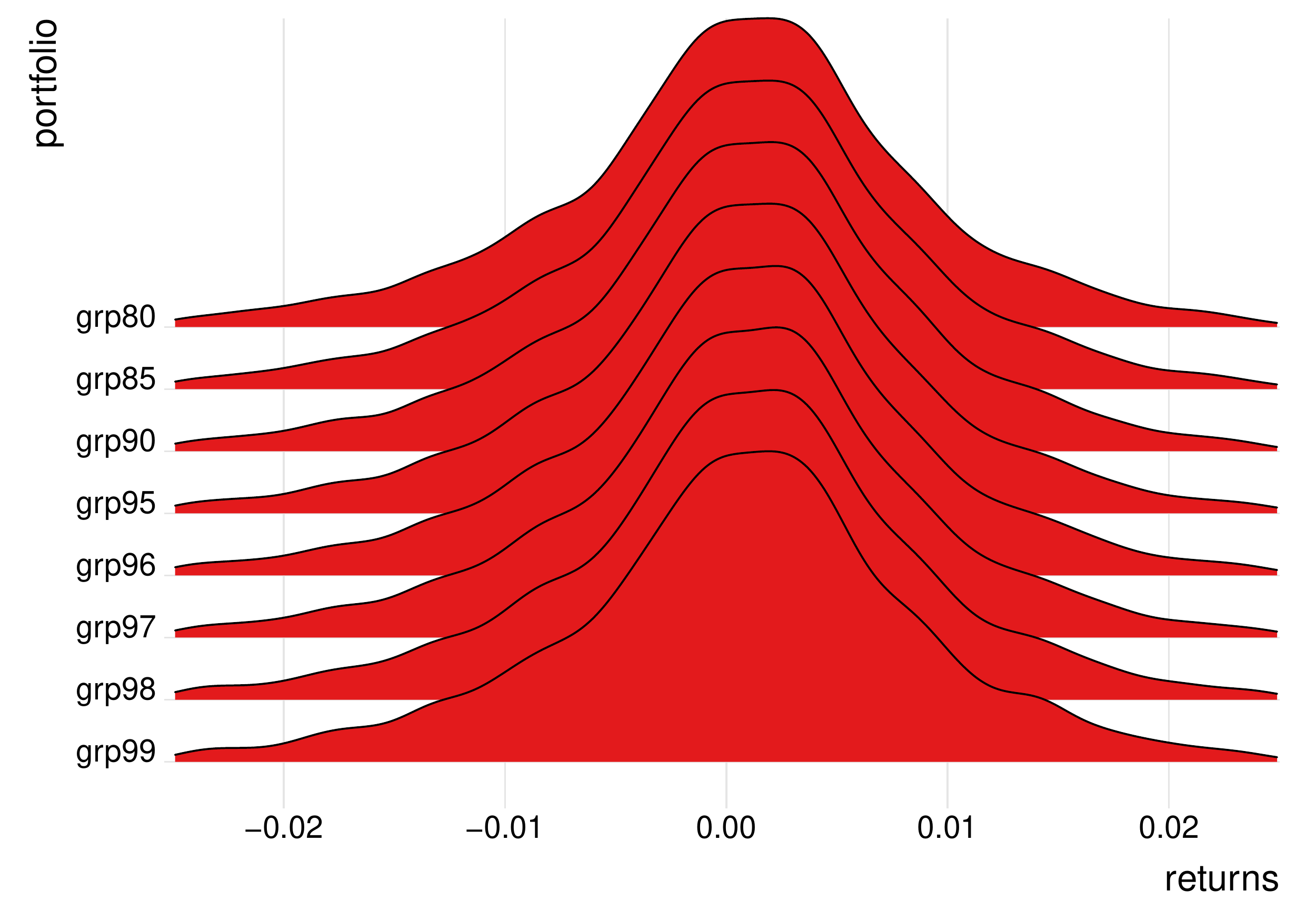} }}%
	\qquad
	\subfloat[\centering \rpa{} portfolios]{{\includegraphics[width=0.45\linewidth]{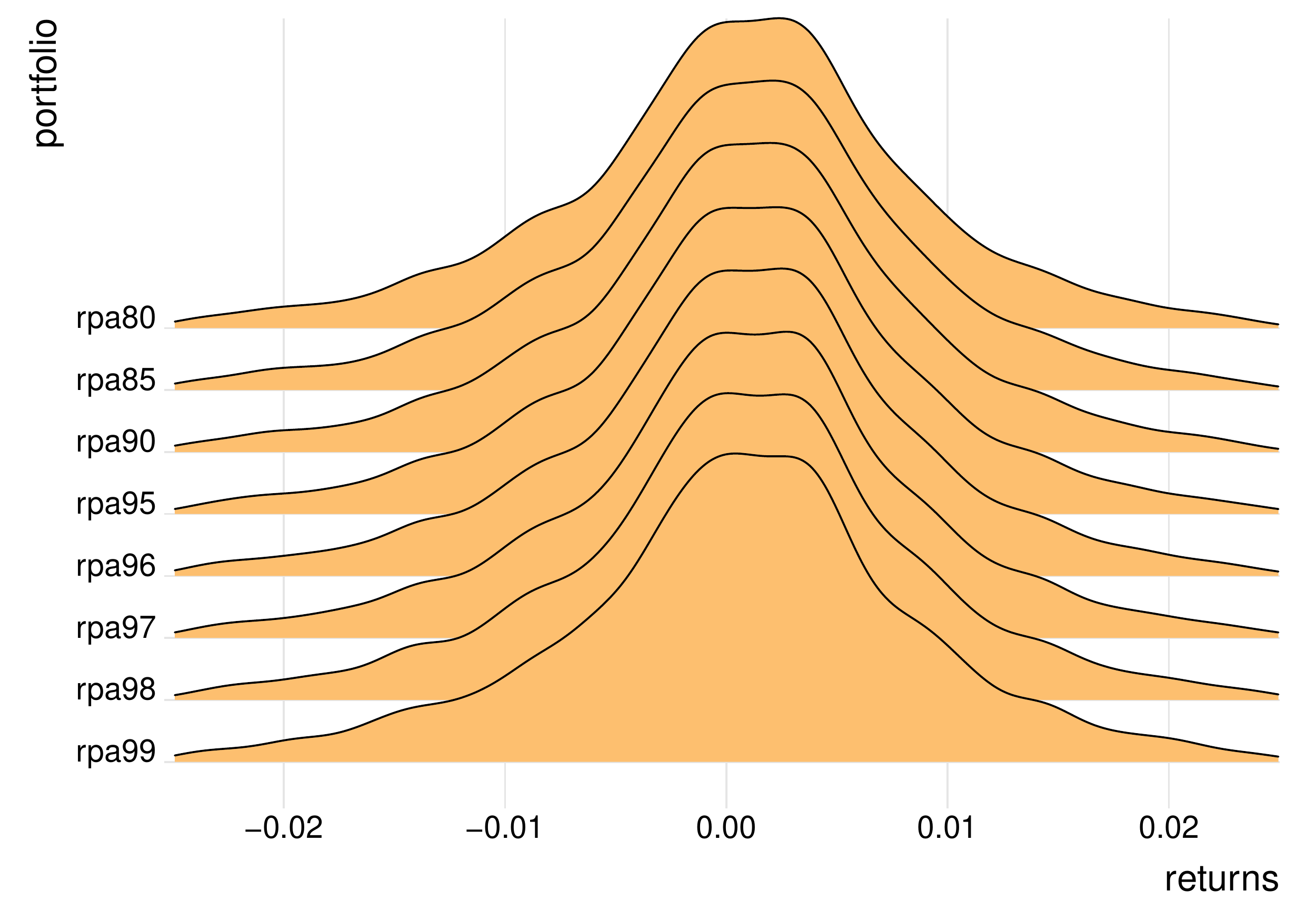} }}%
	\qquad
	\subfloat[\centering \rpb{} portfolios]{{\includegraphics[width=0.4\linewidth]{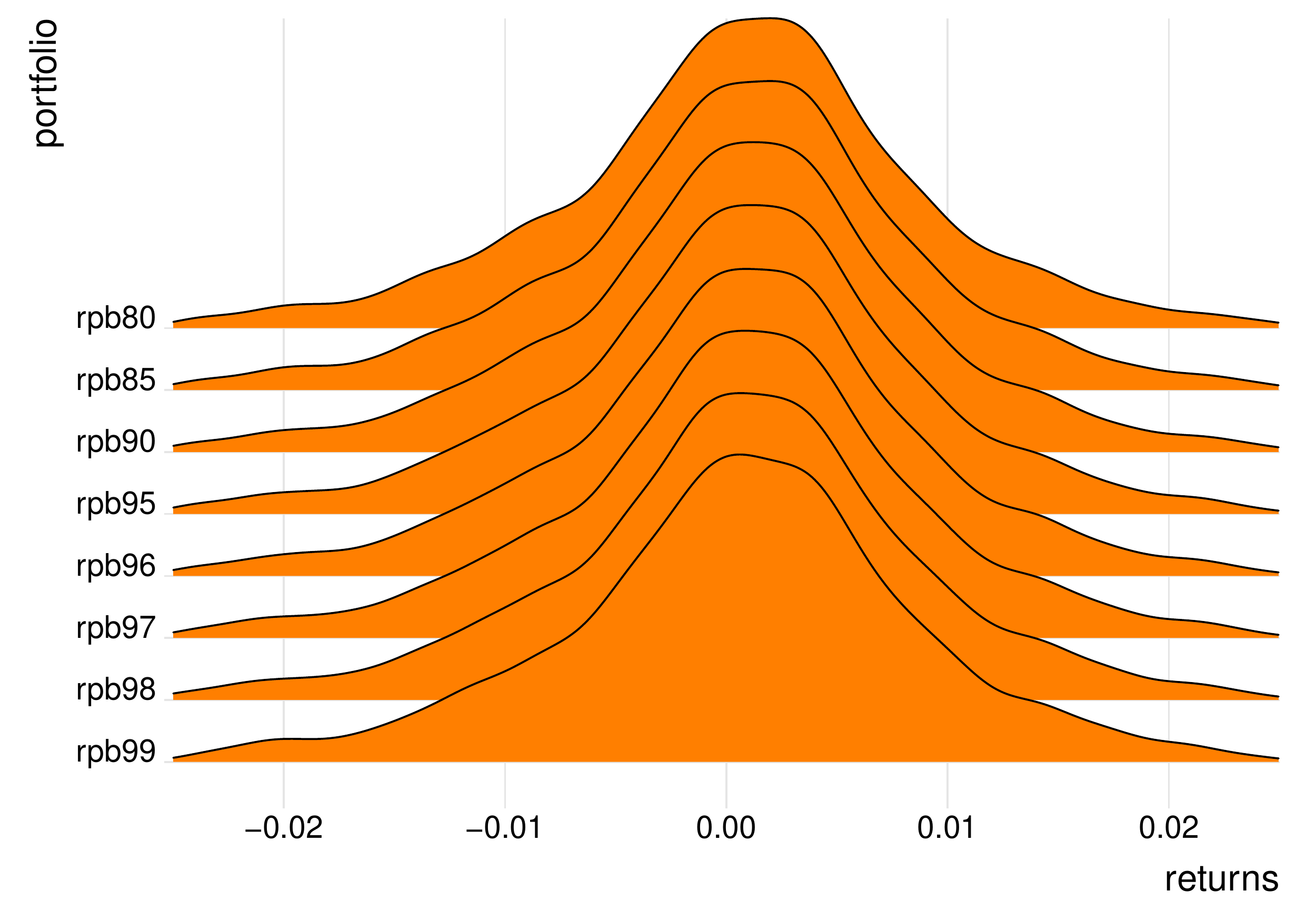} }}%
	\caption{Histogram of the returns of all portfolios considered.}%
	\label{fig:robustness_alpha}%
\end{figure}

\begin{table}[ht]
\centering
\begin{tabular}{lrrrrrrrr}
  \toprule\toprule
 & \msr{} & \mmv{} & \gmv{} & \ew{} & \sdrp{}  & \grp{95}{} & \rpa{95}{} & \rpb{95}{} \\ 
  \midrule
  2006 &  0.96 & \cellcolor{colTable0} 4.40 &  0.80 &  2.11 & \cellcolor{colTable2} 3.06 & \cellcolor{colTable1} 3.13 &  2.72 &  2.69 \\ 
  2007 & \cellcolor{colTable0} 2.85 & \cellcolor{colTable1} 1.95 & \cellcolor{colTable2} 1.91 &  1.47 &  1.74 &  1.75 &  1.62 &  1.59 \\ 
  2008 & \cellcolor{colTable1}-1.16 & -1.43 & \cellcolor{colTable0}-0.77 & \cellcolor{colTable2}-1.24 & -1.32 & -1.33 & -1.24 & -1.24 \\ 
  2009 & \cellcolor{colTable0} 2.96 &  1.93 & \cellcolor{colTable1} 2.53 &  1.85 &  1.96 & \cellcolor{colTable2} 2.03 &  1.72 &  1.72 \\ 
  2010 & -0.36 &  0.18 & -0.16 & \cellcolor{colTable1} 0.45 &  0.39 & \cellcolor{colTable0} 0.47 & \cellcolor{colTable2} 0.40 &  0.38 \\ 
  2011 & -1.56 & \cellcolor{colTable0}-0.55 & -1.78 & \cellcolor{colTable1}-0.65 & -0.69 & \cellcolor{colTable2}-0.68 & -0.75 & -0.73 \\ 
  2012 &  0.51 &  0.87 &  0.09 &  1.13 & \cellcolor{colTable2} 1.16 & \cellcolor{colTable1} 1.17 &  1.15 & \cellcolor{colTable0} 1.24 \\ 
  2013 & \cellcolor{colTable0} 2.83 & \cellcolor{colTable2} 1.94 & \cellcolor{colTable1} 2.78 &  1.82 &  1.81 &  1.86 &  1.59 &  1.69 \\ 
  2014 & \cellcolor{colTable1} 2.17 & \cellcolor{colTable0} 2.85 & \cellcolor{colTable2} 1.87 &  1.33 &  1.73 &  1.70 &  1.71 &  1.72 \\ 
  2015 & \cellcolor{colTable1} 0.68 & \cellcolor{colTable2} 0.36 & \cellcolor{colTable0} 0.72 &  0.31 &  0.33 &  0.34 &  0.33 &  0.35 \\ 
  2016 &  0.58 &  1.01 &  0.69 & \cellcolor{colTable0} 1.10 & \cellcolor{colTable2} 1.04 & \cellcolor{colTable1} 1.08 &  1.04 &  0.98 \\ 
  2017 & \cellcolor{colTable1} 3.27 & \cellcolor{colTable0} 4.03 & \cellcolor{colTable2} 3.26 &  2.93 &  3.08 &  3.09 &  3.05 &  3.04 \\ 
  2018 & \cellcolor{colTable0}-0.55 & \cellcolor{colTable2}-1.02 & -1.02 & -1.04 & -1.05 & -1.09 & \cellcolor{colTable1}-1.00 & -1.05 \\ 
  2019 & \cellcolor{colTable0} 2.92 &  2.33 & \cellcolor{colTable1} 2.89 &  2.63 &  2.62 &  2.63 & \cellcolor{colTable2} 2.64 &  2.62 \\ 
  2020 & \cellcolor{colTable1} 0.99 &  0.15 & \cellcolor{colTable0} 1.32 & \cellcolor{colTable2} 0.39 &  0.37 &  0.37 &  0.34 &  0.33 \\ 
   \bottomrule\bottomrule
\end{tabular}
\caption{Yearly Sharpe ratios for each portfolio. \tikzline[fill=colTable0] highest Sharpe in the year, \tikzline[fill=colTable1]  second largest Sharpe in the year and \tikzline[fill=colTable2] third best Sharpe in the year.}
\label{tbl:stats_table}
\end{table}
In order to evaluate the performance of the portfolios on a yearly basis, Table \ref{tbl:stats_table} presents the Sharpe ratio for all portfolios for each of the 15 years under consideration. Highlighted are the three best Sharpe ratios of the year; where dark green corresponds to the largest Sharpe ratio. In Table \ref{tbl:stats_table} we see a cluster of highlighted cells for the \msr{}, \mmv{} and \gmv{}. Once again, the two DCC-GARCH portfolios present very similar performance throughout the years.

\begin{figure}[h!]
\begin{center}
	\subfloat[\msr{}]{\includegraphics[width = 0.4\textwidth]{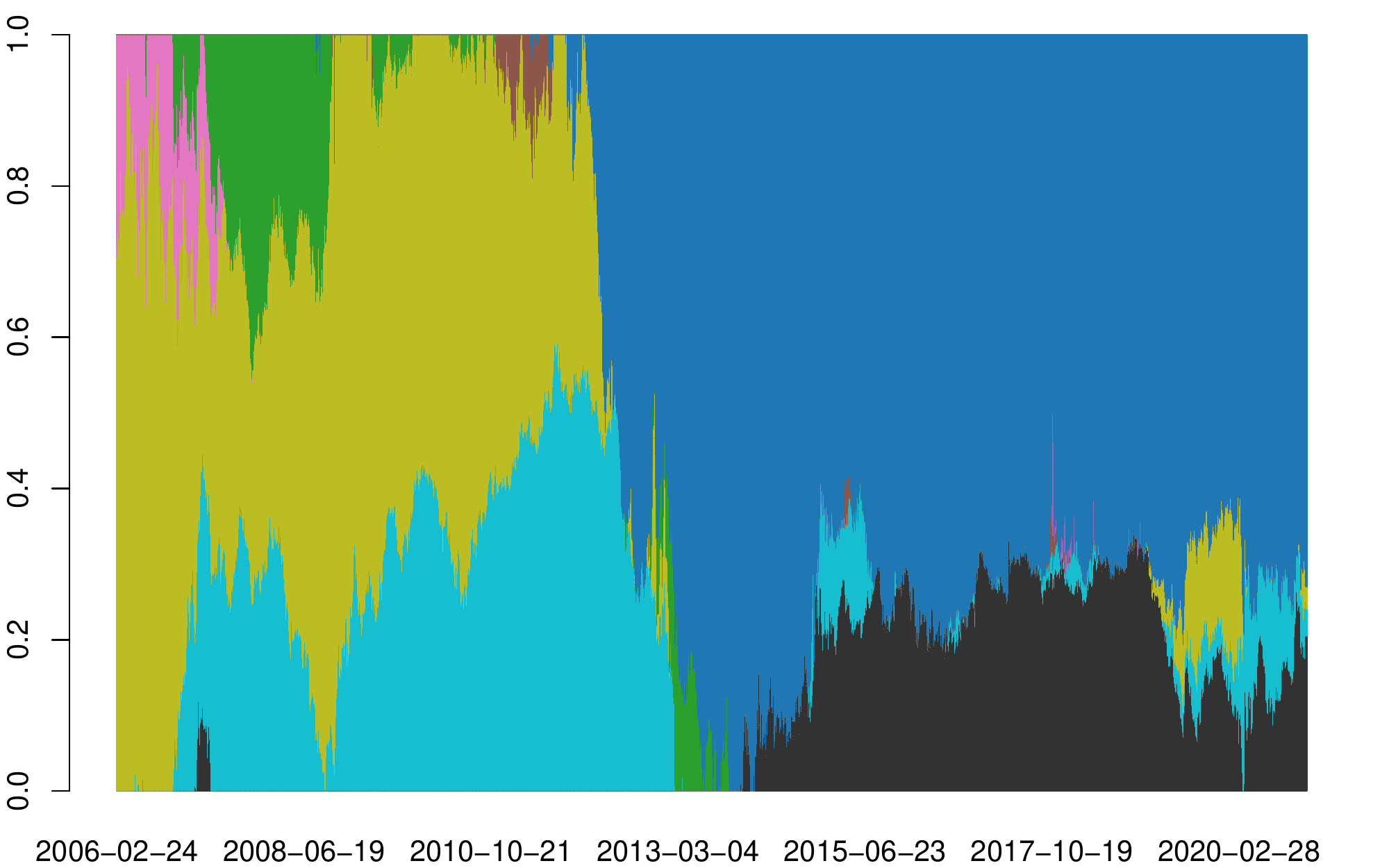}}
	\subfloat[\mmv{}]{\includegraphics[width = 0.4\textwidth]{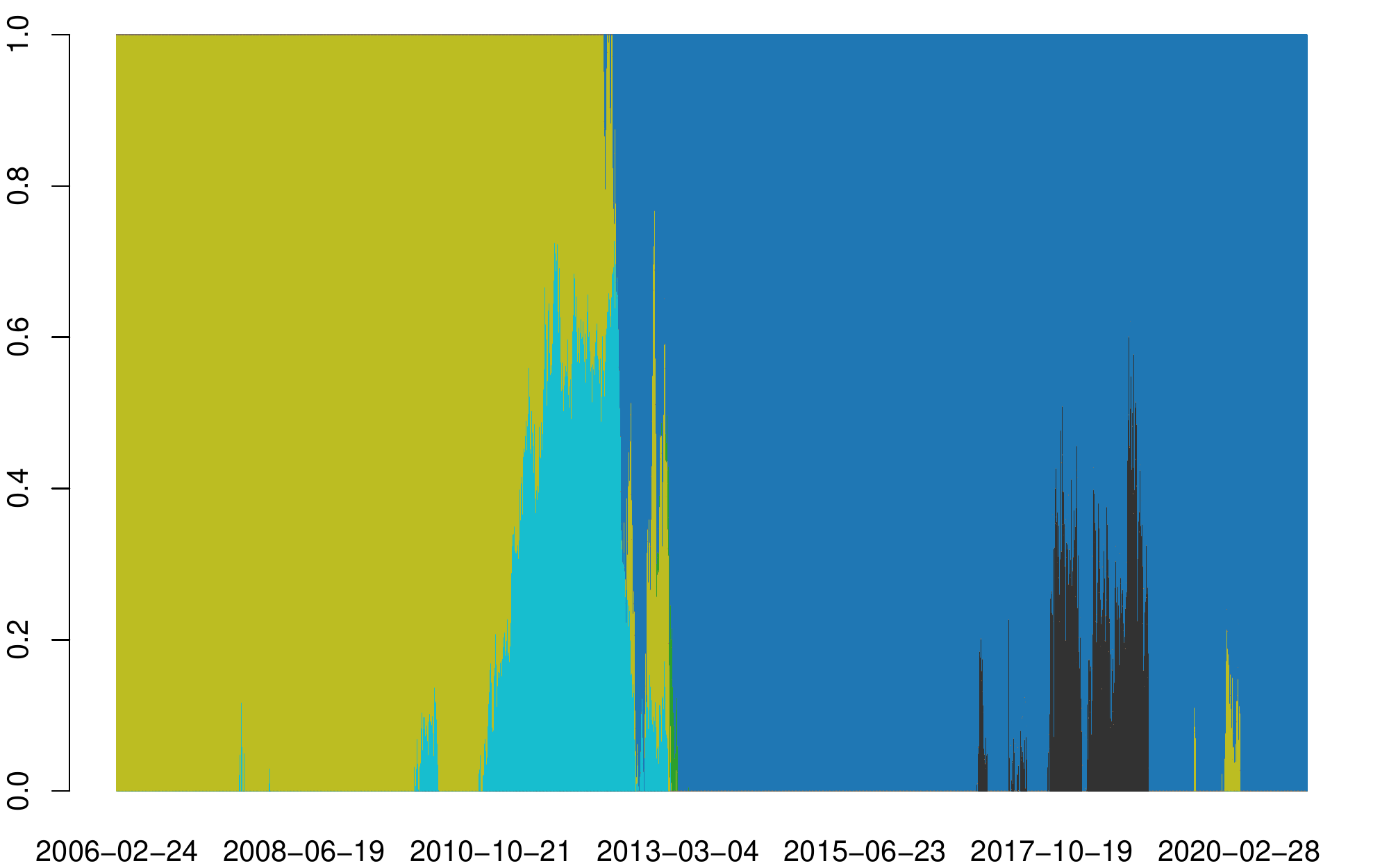}}\\
	\subfloat[\gmv{}]{\includegraphics[width = 0.4\textwidth]{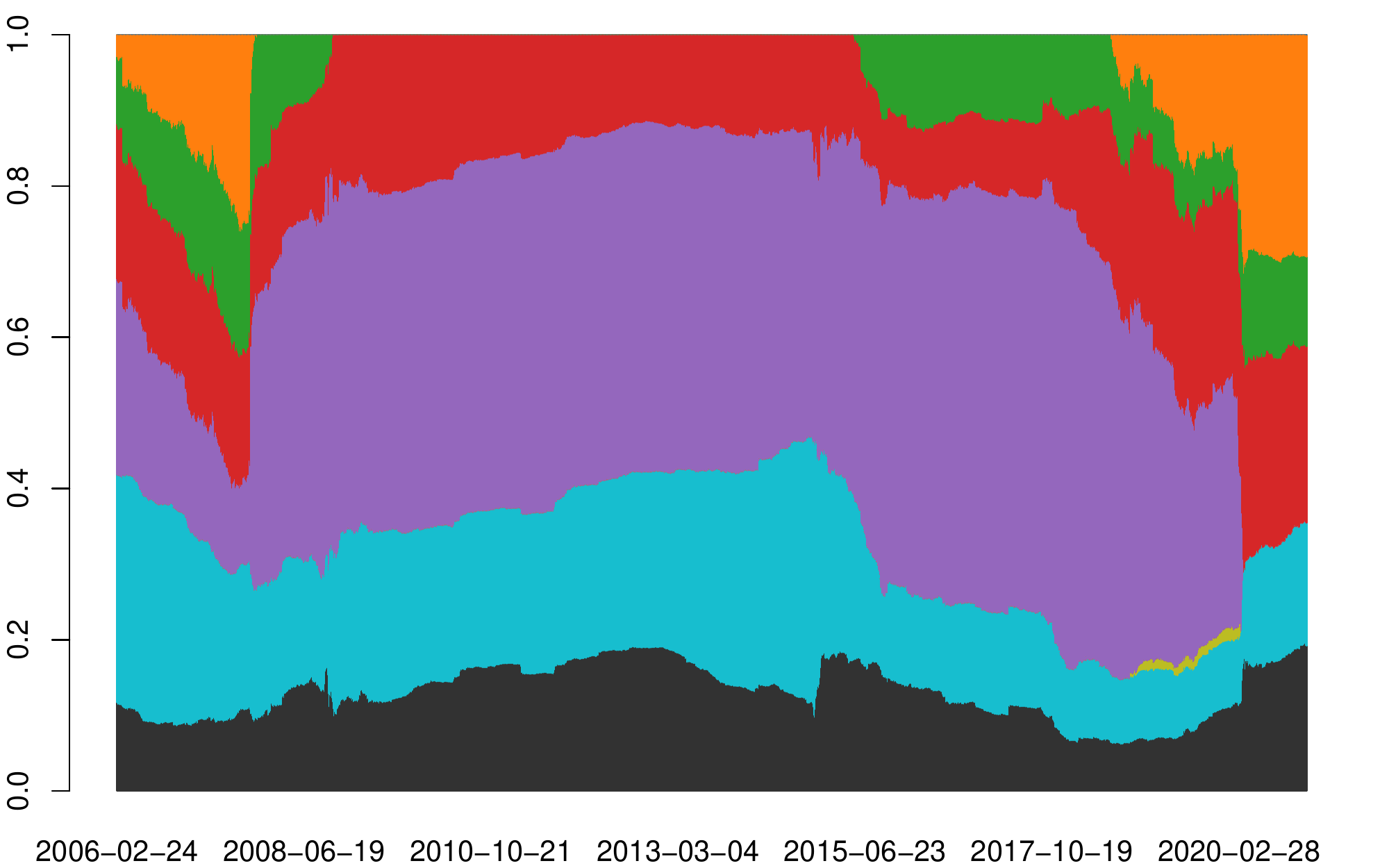}} 
	\subfloat[\ew{}]{\includegraphics[width = 0.4\textwidth]{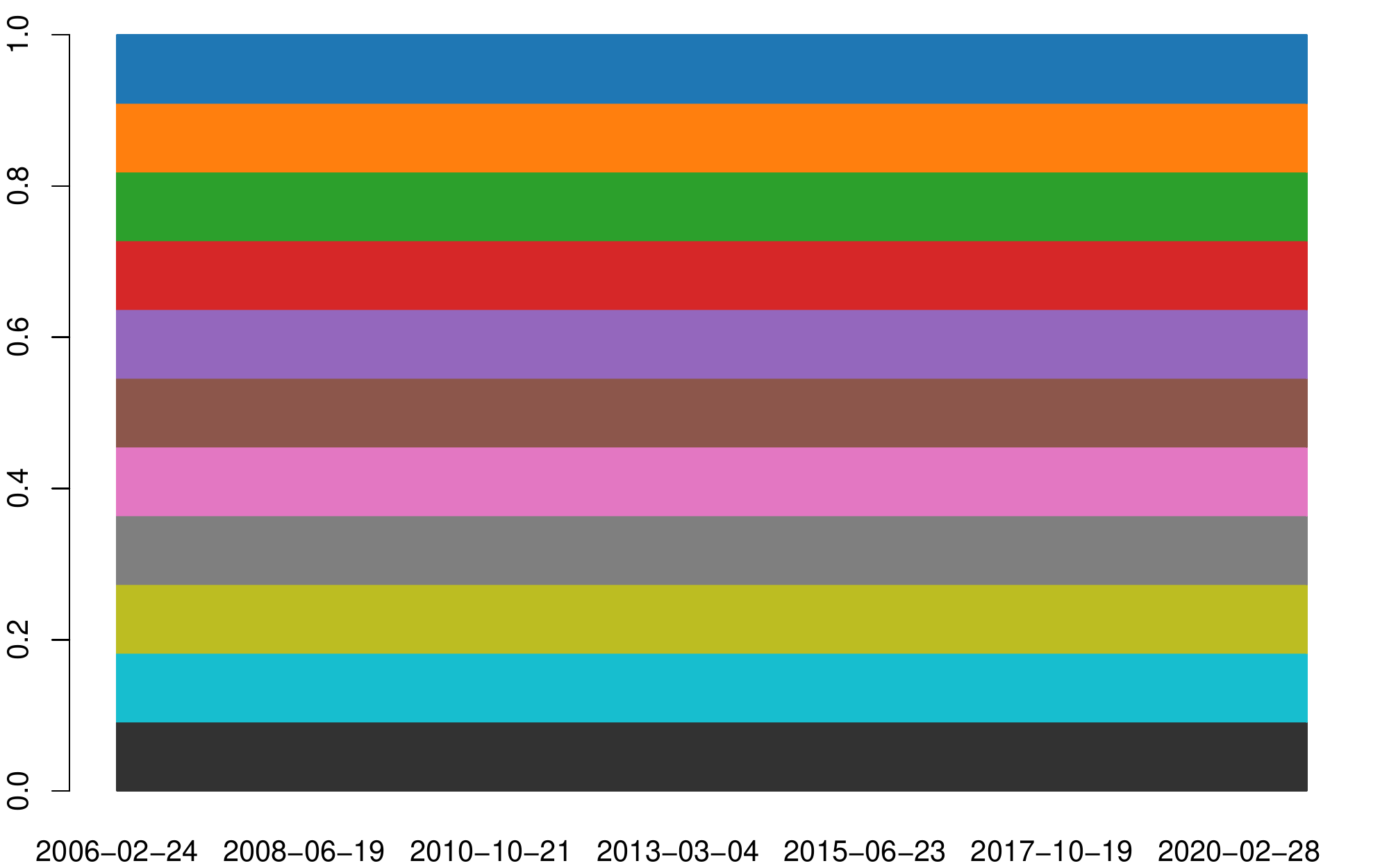}}\\
	\subfloat[\sdrp{}]{\includegraphics[width = 0.4\textwidth]{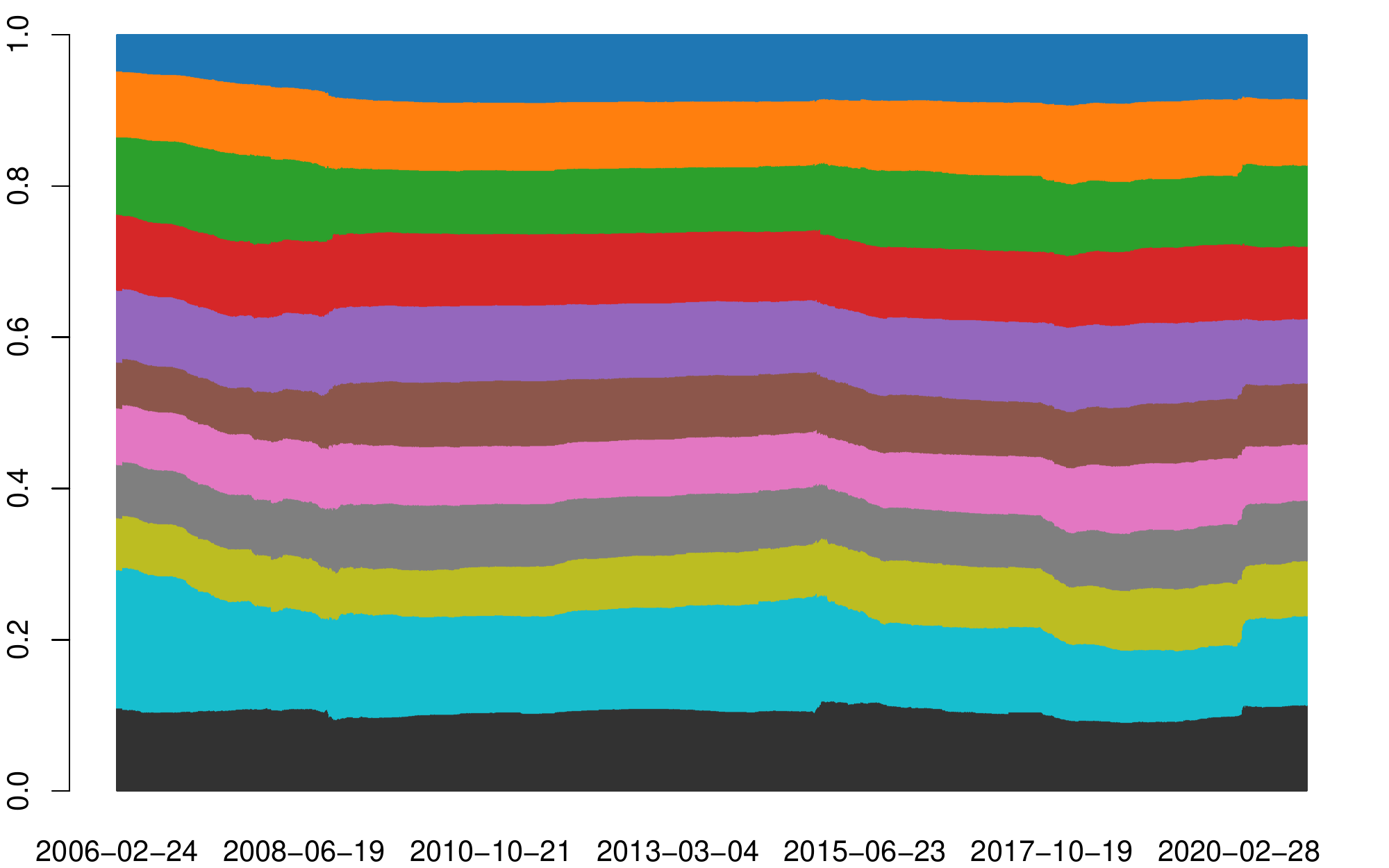}}
	\subfloat[\grp{95}{}]{\includegraphics[width = 0.4\textwidth]{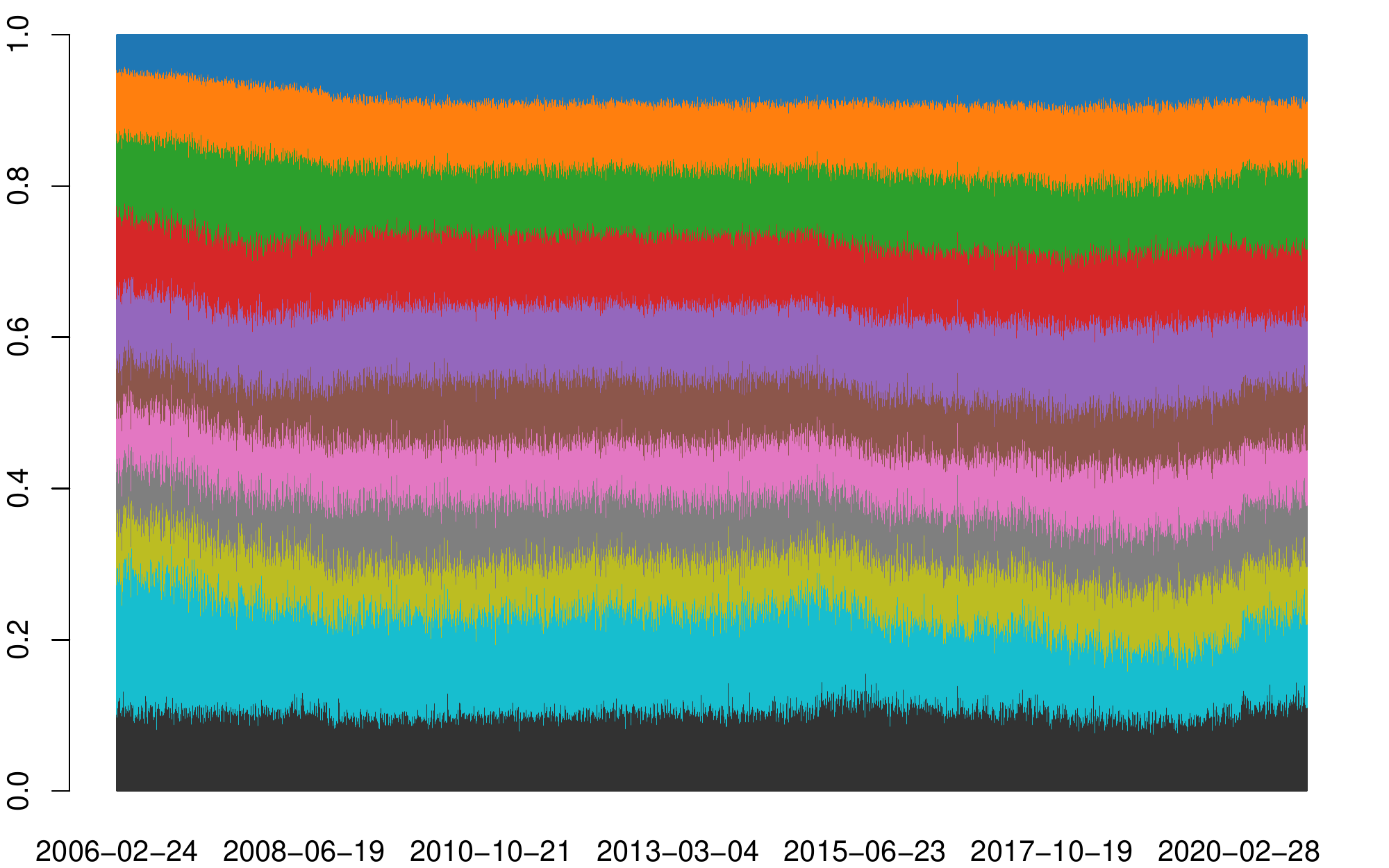}}\\
	\subfloat[\rpa{95}{}]{\includegraphics[width = 0.4\textwidth]{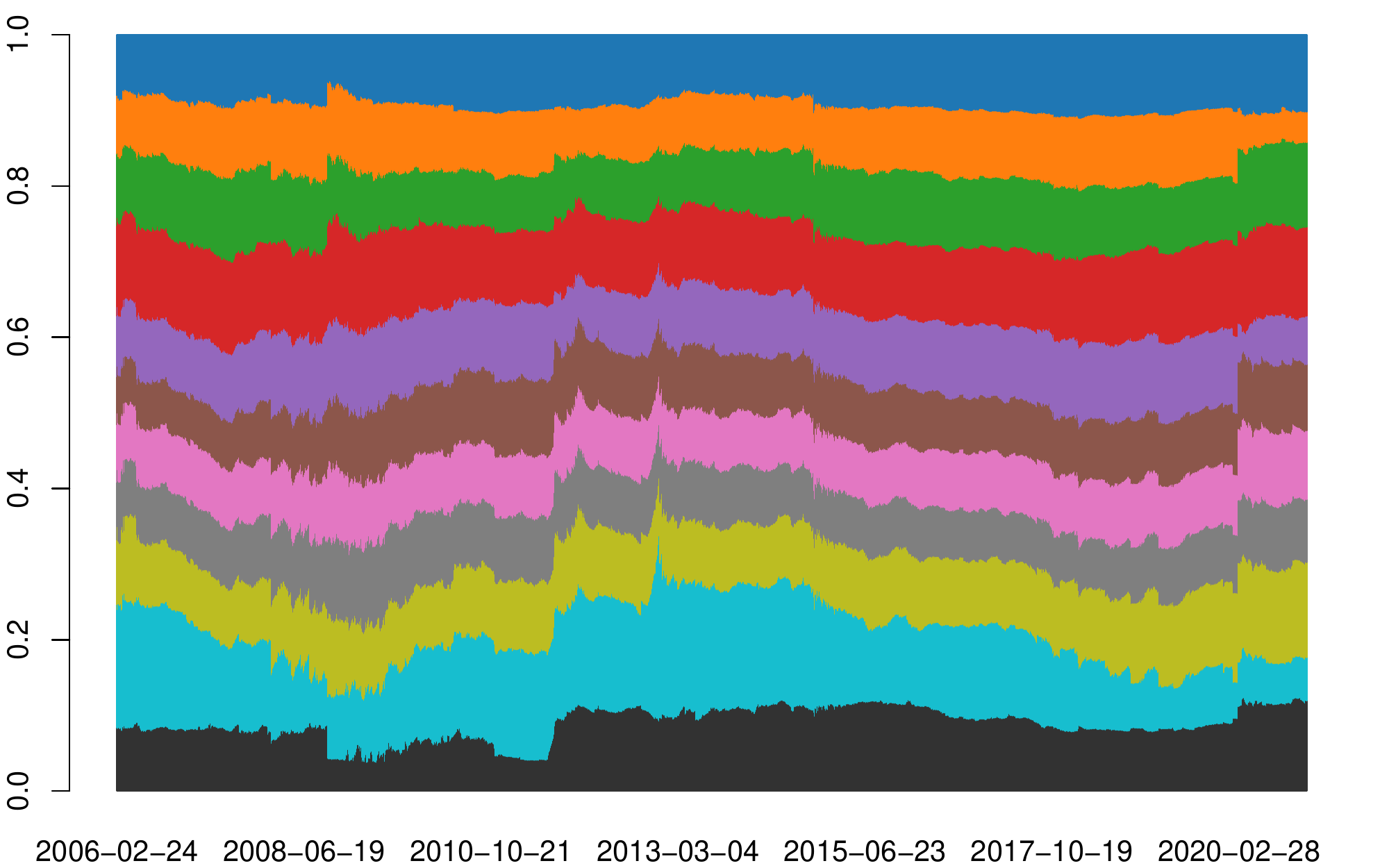}}
	\subfloat[\rpb{95}{}]{\includegraphics[width = 0.4\textwidth]{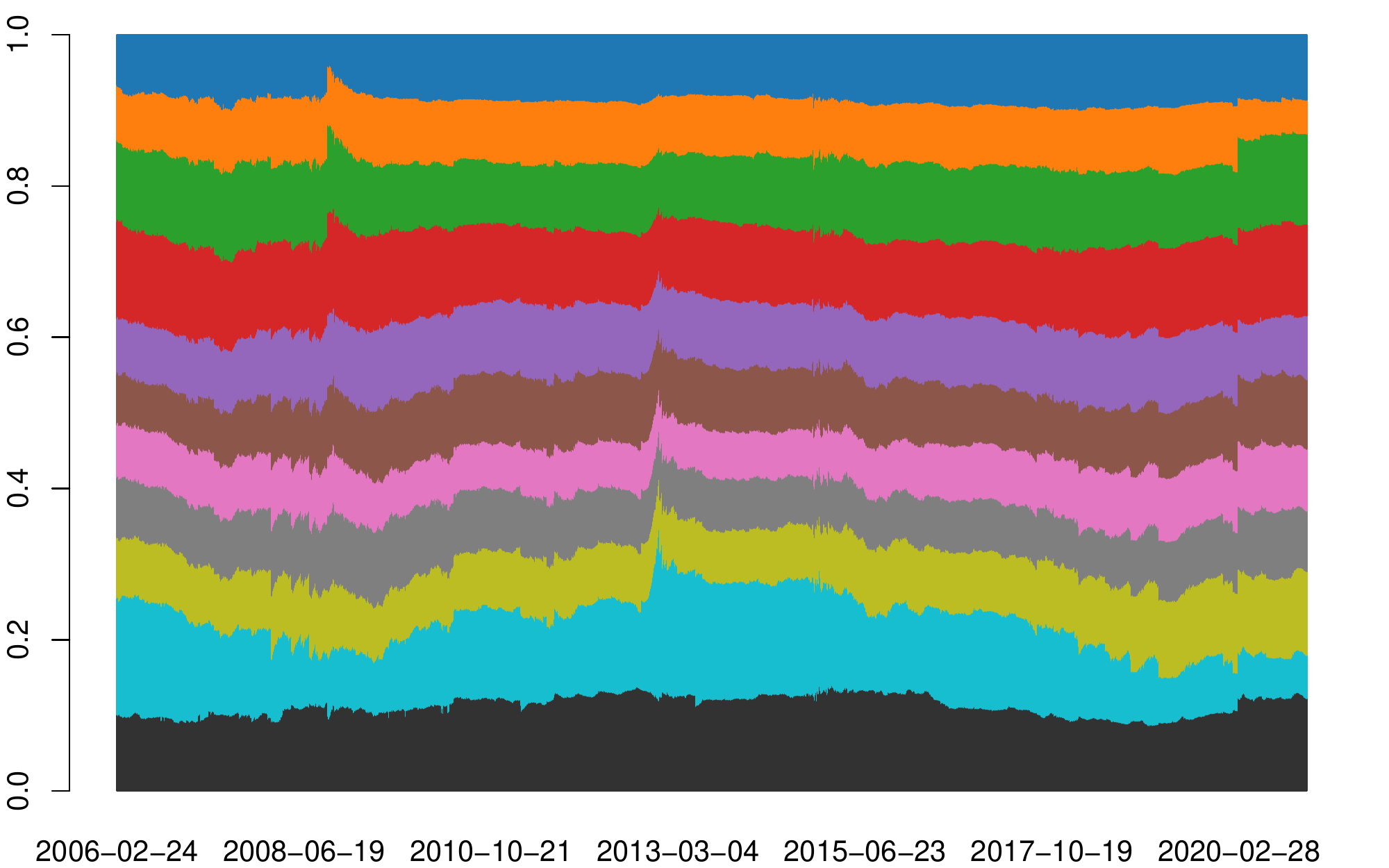}}
\caption{Allocations within each portfolio. \tikzline[fill=colAsset1] NASDAQ 100, 
	\tikzline[fill=colAsset2] S\&P 500, 
	\tikzline[fill=colAsset3] Hang Seng Index, 
	\tikzline[fill=colAsset4] FTSE 100, 
	\tikzline[fill=colAsset5] Dow Jones Industrial Average, 
	\tikzline[fill=colAsset6] DAX Performance-Index, 
	\tikzline[fill=colAsset7] Russell 2000, 
	\tikzline[fill=colAsset8] CAC 40, 
	\tikzline[fill=colAsset9] Ibovespa, 
	\tikzline[fill=colAsset10] SSE Composite Index, 
	\tikzline[fill=colAsset11] Nikkei 225}
	\label{fig:barplots}
\end{center}
\end{figure}

The weights of all 11 assets for each portfolio is presented in Figure \ref{fig:barplots}. As suggested by Figure \ref{fig:all_assets_pairs}, \msr{}, \mmv{}, and \gmv{} are the most dissimilar ones. From the weights of the \mmv{} portfolio, we see that the  \mmv{}, for some periods of time, is fully invested in one single asset or, at most, two. A similar, although not so extreme, behaviour is observed for the \msr{} and \gmv{} portfolios. Both \mmv{} and \msr{} are strongly invested on NASDAQ 100 between 2013 and 2020. During this period, NASDAQ 100 presented extremely high returns, compared to the other assets in the sample, leading to similar behaviour in the portfolios.

\begin{figure}[ht]%
	\centering
	\includegraphics[width = 0.4\textwidth]{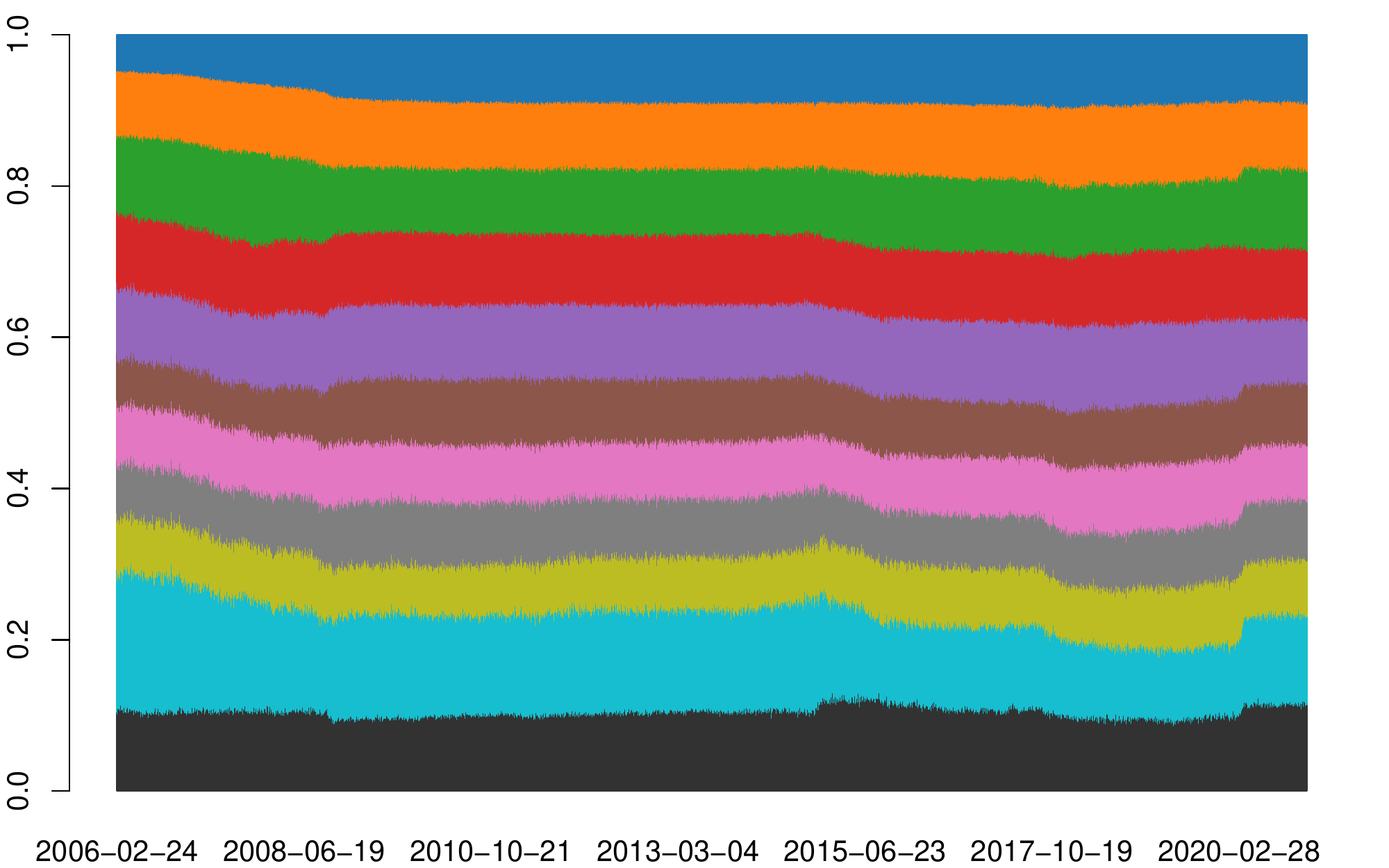}%
	\caption{Weights of the \grp{95}{} portfolio. Differently from Figure \ref{fig:barplots}, here using $N=10\,000$ samples.}%
	\label{fig:convergence_in_N}%
\end{figure}

Within the risk parity strategies, the first noticeable feature is the fact that all portfolios are much more balanced than \msr{}, \mmv{}, and \gmv{}. Due to the nature of the assets under consideration (all with volatilities between 20\% and 30\%), the \sdrp{} is, not surprisingly, similar to \ew{}. As the empirical returns are sufficiently close to zero, the ES is proportional to the standard deviation under the Gaussian hypothesis (see Example \ref{ex:Gaussian-ES}). Therefore, the allocations under the \grp{95}{} should converge to those of \sdrp{} when $N \rightarrow +\infty$. This convergence is studied in Figure \ref{fig:convergence_in_N}, where we present the \grp{95}{} portfolio generated using $N=10\,000$, instead of $N= 1\,000$ (Figure \ref{fig:barplots}) samples in Algorithm \ref{alg:ES_cp}. As expected, the \grp{}{} computed with larger sample size $N$ is closer to \sdrp{}. 

The portfolio weights under the DCC-GARCH models (\rpa{95}{} and \rpb{95}{}) are very similar, while the weights under the DCC-GJR-GARCH specification (\rpb{}) appear to be more robust. This suggests that a more flexible model may be needed to take full advantage of the risk parity strategy using ES. Overall, the major movements observed in the \sdrp{} weights are similar to those of both risk parity portfolios. For example, all risk parity portfolios start with a large proportion of wealth invested in the SSE Composite Index and the weights in that asset in the portfolios decrease steadily until the global financial crisis in 2008, after which the weights start to increase again. Differently from the \msr{} and the \mmv{}, a major change in the weights of the \rpa{95}{} and \rpb{95}{} portfolios is observed during the Covid-19 Pandemic.

It must be stressed that the behaviour of the risk parity portfolios constructed using the ES as a risk measure and complex multivariate dynamics for the returns is only currently possible due to the algorithms developed in this paper (Section \ref{sec:algorithms}).    

\section{Conclusions} \label{sec:conclusion}
We presented three optimisation algorithms
tailored to the problem of constructing risk budgeting portfolios
for coherent risk measures.
One key benefits of the proposed algorithms
is the fact that no analytical expressions for the densities of the assets' returns is needed,
indeed both our cutting planes algorithms and stochastic gradient descent algorithms rely only on \textit{scenarios} of returns.
We compare our cutting planes algorithm with standard convex optimisation solvers
and show that our algorithm is significantly faster, particularly for increasing asset dimension and number of scenarios.
Moreover, for a distortion risk measures with continuous distortion function and the Entropic VaR,
only our cutting planes algorithm converges.

We apply our proposed algorithm to construct portfolios based on 11 equity indices
and analyse the performance of the risk parity portfolios compared to other investment strategies.
In this particular application we see that some of the risk parity portfolios
are empirically similar to a portfolio with equal weights;
a conclusion to date only available for risk parity portfolios using the standard deviation as a risk measure.    
\section*{Acknowledgements}
SP would like to acknowledge support from the Natural Sciences and Engineering Research Council of Canada (grants DGECR-2020-00333, RGPIN-2020-04289).

\bibliographystyle{apalike}
\bibliography{references.bib}

\begin{thebibliography}{}

\bibitem[Ahmadi-Javid, 2012]{ahmadi2012entropic}
Ahmadi-Javid, A. (2012).
\newblock Entropic value-at-risk: A new coherent risk measure.
\newblock {\em Journal of Optimization Theory and Applications},
  155:1105--1123.

\bibitem[Anis and Kwon, 2022]{anis2022EJOR}
Anis, H.~T. and Kwon, R.~H. (2022).
\newblock Cardinality-constrained risk parity portfolios.
\newblock {\em European Journal of Operational Research}, 302(1):392--402.

\bibitem[Artzner et~al., 1999]{Artzner1999MF}
Artzner, P., Delbaen, F., Eber, J.-M., and Heath, D. (1999).
\newblock Coherent measures of risk.
\newblock {\em Mathematical Finance}, 9(3):203--228.

\bibitem[Bai et~al., 2016]{bai2016least}
Bai, X., Scheinberg, K., and Tutuncu, R. (2016).
\newblock Least-squares approach to risk parity in portfolio selection.
\newblock {\em Quantitative Finance}, 16(3):357--376.

\bibitem[Bellini et~al., 2021]{Bellini2021EJOR}
Bellini, F., Cesarone, F., Colombo, C., and Tardella, F. (2021).
\newblock Risk parity with expectiles.
\newblock {\em European Journal of Operational Research}, 291(3):1149--1163.

\bibitem[Bollerslev, 1986]{bollerslev1986generalized}
Bollerslev, T. (1986).
\newblock Generalized autoregressive conditional heteroskedasticity.
\newblock {\em Journal of Econometrics}, 31(3):307--327.

\bibitem[Boudt et~al., 2013]{Boudt2012JR}
Boudt, K., Carl, P., and Peterson, B.~G. (2013).
\newblock Asset allocation with conditional {V}alue-at-{R}isk budgets.
\newblock {\em Journal of Risk}, 15(3):39--68.

\bibitem[Bruder and Roncalli, 2012]{bruder2012managing}
Bruder, B. and Roncalli, T. (2012).
\newblock Managing risk exposures using the risk budgeting approach.
\newblock {\em Available at SSRN 2009778}.

\bibitem[Cesarone and Tardella, 2017]{Cesarone2017JGO}
Cesarone, F. and Tardella, F. (2017).
\newblock Equal risk bounding is better than risk parity for portfolio
  selection.
\newblock {\em Journal of Global Optimization}, 68(2):439--461.

\bibitem[Chaves et~al., 2012]{chaves2012efficient}
Chaves, D., Hsu, J., Li, F., and Shakernia, O. (2012).
\newblock Efficient algorithms for computing riskparity portfolio weights.
\newblock {\em The Journal of Investing}, 21(3):150--163.

\bibitem[Christoffersen et~al., 1998]{christoffersen1998horizon}
Christoffersen, P., Diebold, F.~X., and Schuermann, T. (1998).
\newblock Horizon problems and extreme events in financial risk management.
\newblock {\em Economic Policy Review}, 4(3):98--16.

\bibitem[da~Costa, 2023]{RiskBudgeting.jl}
da~Costa, B. F.~P. (2023).
\newblock {RiskBudgeting}: a package for risk budgeting portfolios.
\newblock \url{https://github.com/bfpc/RiskBudgeting.jl}.

\bibitem[Darolles et~al., 2015]{Darolles2015book_chapter}
Darolles, S., Gourieroux, C., and Jay, E. (2015).
\newblock Robust portfolio allocation with systematic risk contribution
  restrictions.
\newblock In {\em Risk-based and Factor Investing}, pages 123--146. Elsevier.

\bibitem[de~M.~Cardoso and Palomar, 2021]{cardoso2021riskparity}
de~M.~Cardoso, J.~V. and Palomar, D.~P. (2021).
\newblock {\em {riskParityPortfolio: Design of Risk Parity Portfolios}}.
\newblock R package version 0.2.2.

\bibitem[Dentcheva et~al., 2010]{Dentcheva2010Renyi}
Dentcheva, D., Penev, S., and Ruszczy{\'n}ski, A. (2010).
\newblock Kusuoka representation of higher order dual risk measures.
\newblock {\em Annals of Operations Research}, 181(1):325--335.

\bibitem[Dhaene et~al., 2012]{Dhaene2012EAJ}
Dhaene, J., Kukush, A., Linders, D., and Tang, Q. (2012).
\newblock Remarks on quantiles and distortion risk measures.
\newblock {\em European Actuarial Journal}, 2(2):319--328.

\bibitem[Dunning et~al., 2017]{JuMP2017}
Dunning, I., Huchette, J., and Lubin, M. (2017).
\newblock Jump: A modeling language for mathematical optimization.
\newblock {\em SIAM Review}, 59(2):295--320.

\bibitem[Engle, 2002]{engle2002dynamic}
Engle, R. (2002).
\newblock Dynamic conditional correlation: A simple class of multivariate
  generalized autoregressive conditional heteroskedasticity models.
\newblock {\em Journal of Business \& Economic Statistics}, 20(3):339--350.

\bibitem[Feng and Palomar, 2015]{feng2015scrip}
Feng, Y. and Palomar, D.~P. (2015).
\newblock Scrip: Successive convex optimization methods for risk parity
  portfolio design.
\newblock {\em IEEE Transactions on Signal Processing}, 63(19):5285--5300.

\bibitem[Ghalanos, 2019]{rmgarch}
Ghalanos, A. (2019).
\newblock {\em rmgarch: Multivariate GARCH models.}
\newblock R package version 1.3-7.

\bibitem[Glosten et~al., 1993]{glosten1993relation}
Glosten, L.~R., Jagannathan, R., and Runkle, D.~E. (1993).
\newblock On the relation between the expected value and the volatility of the
  nominal excess return on stocks.
\newblock {\em The Journal of Finance}, 48(5):1779--1801.

\bibitem[Grechuk, 2023]{Grechuk2023EJOR}
Grechuk, B. (2023).
\newblock Extended gradient of convex function and capital allocation.
\newblock {\em European Journal of Operational Research}, 305(1):429--437.

\bibitem[Griveau-Billion et~al., 2013]{griveau2013fast}
Griveau-Billion, T., Richard, J.-C., and Roncalli, T. (2013).
\newblock A fast algorithm for computing high-dimensional risk parity
  portfolios.
\newblock {\em Available at SSRN 2325255}.

\bibitem[Hong and Liu, 2009]{hong2009simulating}
Hong, L.~J. and Liu, G. (2009).
\newblock Simulating sensitivities of conditional value at risk.
\newblock {\em Management Science}, 55(2):281--293.

\bibitem[Jurczenko and Teiletche, 2019]{jurczenko2019expected}
Jurczenko, E. and Teiletche, J. (2019).
\newblock Expected shortfall asset allocation: A multi-dimensional
  risk-budgeting framework.
\newblock {\em The Journal of Alternative Investments}, 22(2):7--22.

\bibitem[Krokhmal, 2007]{Krokhmal2007higher}
Krokhmal, P.~A. (2007).
\newblock Higher moment coherent risk measures.
\newblock {\em Quantitative Finance}, 7:373--387.

\bibitem[Ledoit and Wolf, 2017]{ledoit2017numerical}
Ledoit, O. and Wolf, M. (2017).
\newblock Numerical implementation of the quest function.
\newblock {\em Computational Statistics \& Data Analysis}, 115:199--223.

\bibitem[Li et~al., 2021]{Li2021EJOR}
Li, X., Uysal, A.~S., and Mulvey, J.~M. (2021).
\newblock Multi-period portfolio optimization using model predictive control
  with mean-variance and risk parity frameworks.
\newblock {\em European Journal of Operational Research}.

\bibitem[Maillard et~al., 2010]{Maillard2010JPM}
Maillard, S., Roncalli, T., and Te{\"\i}letche, J. (2010).
\newblock The properties of equally weighted risk contribution portfolios.
\newblock {\em The Journal of Portfolio Management}, 36(4):60--70.

\bibitem[Markowitz, 1952]{markowitz1952portfolio}
Markowitz, H. (1952).
\newblock Portfolio selection.
\newblock {\em Journal of Finance}, 7(1):77--91.

\bibitem[Mausser and Romanko, 2018]{mausser2018long}
Mausser, H. and Romanko, O. (2018).
\newblock Long-only equal risk contribution portfolios for {CVaR} under
  discrete distributions.
\newblock {\em Quantitative Finance}, 18(11):1927--1945.

\bibitem[Meyer and Quell, 2020]{risk2020meyer}
Meyer, C. and Quell, P. (2020).
\newblock {\em Risk Model Validation}.
\newblock Risk books.

\bibitem[Mosek, 2022]{mosek}
Mosek (2022).
\newblock {Mosek.jl}: Interface to the mosek solver in julia.
\newblock \url{https://github.com/MOSEK/Mosek.jl}.

\bibitem[O'Donoghue et~al., 2016]{ocpb:16}
O'Donoghue, B., Chu, E., Parikh, N., and Boyd, S. (2016).
\newblock Conic optimization via operator splitting and homogeneous self-dual
  embedding.
\newblock {\em Journal of Optimization Theory and Applications},
  169(3):1042--1068.

\bibitem[Perlin, 2020]{BatchGetSymbols}
Perlin, M. (2020).
\newblock {\em BatchGetSymbols: Downloads and Organizes Financial Data for
  Multiple Tickers}.
\newblock R package version 2.6.1.

\bibitem[Pesenti et~al., 2021]{Pesenti2021Cascade}
Pesenti, S.~M., Millossovich, P., and Tsanakas, A. (2021).
\newblock Cascade sensitivity measures.
\newblock {\em Risk Analysis}, 41:2392--2414.

\bibitem[Pr{\'e}kopa, 1995]{prekopa1995stochastic}
Pr{\'e}kopa, A. (1995).
\newblock {\em Stochastic programming}, volume 324.
\newblock Kluwer Academic Publishers.

\bibitem[Qian, 2005]{Qian2005RiskParity}
Qian, E. (2005).
\newblock Risk parity portfolios: Efficient portfolios through true
  diversification.
\newblock {\em Panagora Asset Management}.

\bibitem[Ramprasad, 2016]{nlshrink}
Ramprasad, P. (2016).
\newblock {\em nlshrink: Non-Linear Shrinkage Estimation of Population
  Eigenvalues and Covariance Matrices}.
\newblock R package version 1.0.1.

\bibitem[Rockafellar and Uryasev, 2000]{Rockafellar-Uryasev}
Rockafellar, R.~T. and Uryasev, S. (2000).
\newblock Optimization of conditional {V}alue-at-{R}isk.
\newblock {\em Journal of Risk}, 2:21--42.

\bibitem[Roncalli, 2013]{Roncalli2013Book}
Roncalli, T. (2013).
\newblock {\em Introduction to risk parity and budgeting}.
\newblock CRC Press.

\bibitem[Shapiro et~al., 2009]{shapiro2009lectures}
Shapiro, A., Dentcheva, D., and Ruszczy{\'n}ski, A. (2009).
\newblock {\em Lectures on stochastic programming: modeling and theory}.
\newblock SIAM.

\bibitem[Sharpe, 1966]{sharpe1966mutual}
Sharpe, W.~F. (1966).
\newblock Mutual fund performance.
\newblock {\em The Journal of Business}, 39(1):119--138.

\bibitem[Spinu, 2013]{spinu2013algorithm}
Spinu, F. (2013).
\newblock An algorithm for computing risk parity weights.
\newblock {\em Available at SSRN 2297383}.

\bibitem[Tasche, 1999]{Tasche99}
Tasche, D. (1999).
\newblock Risk contributions and performance measurement.
\newblock {\em Report of the Lehrstuhl f{\"u}r mathematische Statistik, TU
  M{\"u}nchen}.

\bibitem[Tsanakas, 2009]{Tsanakas2009IME}
Tsanakas, A. (2009).
\newblock To split or not to split: Capital allocation with convex risk
  measures.
\newblock {\em Insurance: Mathematics and Economics}, 44(2):268--277.

\bibitem[Tsanakas and Millossovich, 2016]{Tsanakas2016RA}
Tsanakas, A. and Millossovich, P. (2016).
\newblock Sensitivity analysis using risk measures.
\newblock {\em Risk Analysis}, 36(1):30--48.

\bibitem[Udell et~al., 2014]{convexjl}
Udell, M., Mohan, K., Zeng, D., Hong, J., Diamond, S., and Boyd, S. (2014).
\newblock Convex optimization in {J}ulia.
\newblock {\em SC14 Workshop on High Performance Technical Computing in Dynamic
  Languages}.

\bibitem[Wächter and Biegler, 2006]{ipopt}
Wächter, A. and Biegler, L.~T. (2006).
\newblock On the implementation of a primal-dual interior point filter line
  search algorithm for large-scale nonlinear programming.
\newblock {\em Mathematical Programming}, 106:25--57.

\end{thebibliography}


\clearpage

\appendix 

\section{Numerical studies} \label{sec:appendix_numerical_studies}
In this section we provide additional plots for the numerical exercises of Section \ref{sec:numerical_studies}. Figures \ref{fig:cvar_gaussian} and \ref{fig:cvar_t} complement, respectively, the multivariate Gaussian and Student $t$ analysis for the ES with $\alpha= 0.95$ and $\alpha =0.99$. Figures \ref{fig:entropicVaR_dim_time_TRUE} and \ref{fig:distortion_dim_time_TRUE} show  the run times for the multivariate Student $t$ for the Entropic VaR and distortion risk measures, respectively.

\begin{figure}[!htp]
\centering
\subfloat[Multivariate Gaussian model with $\alpha = 0.95$.]{%
  \includegraphics[width = 0.8\textwidth]{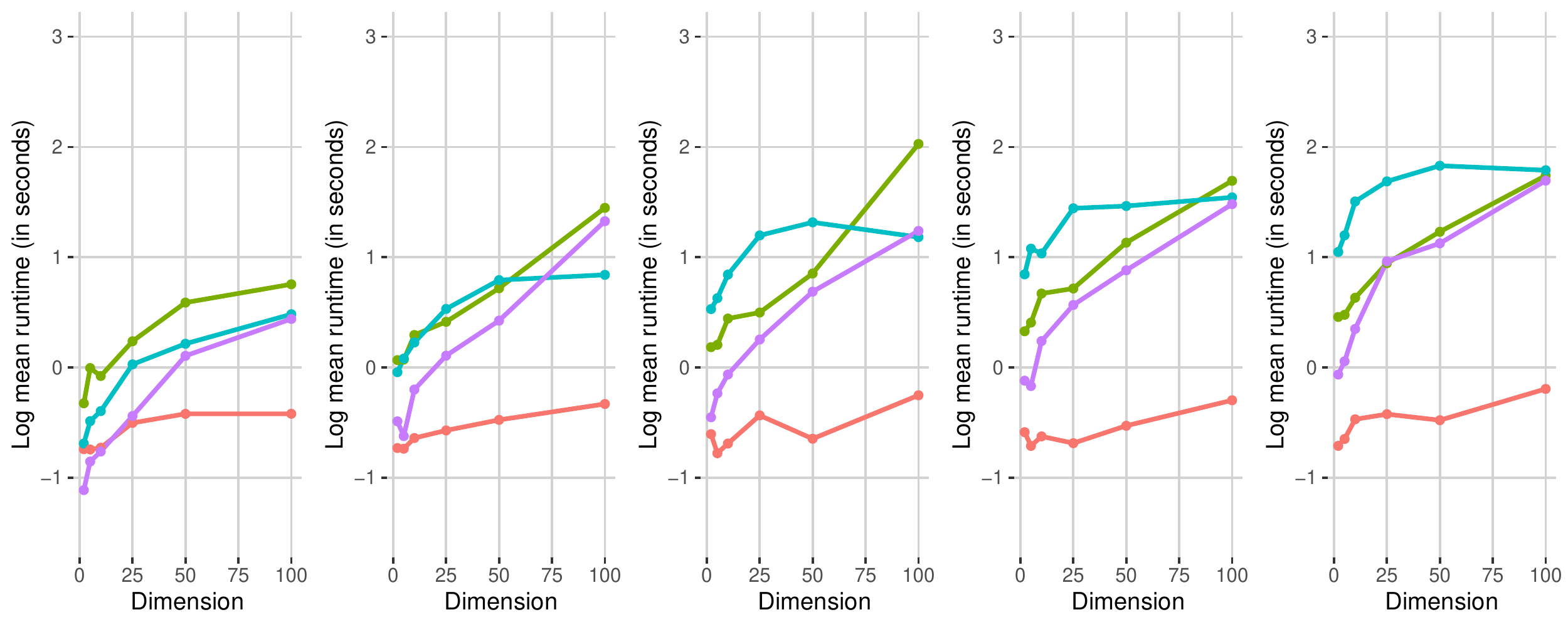}%
}\quad
\subfloat[Multivariate Gaussian model with $\alpha = 0.99$.]{%
  \includegraphics[width = 0.8\textwidth]{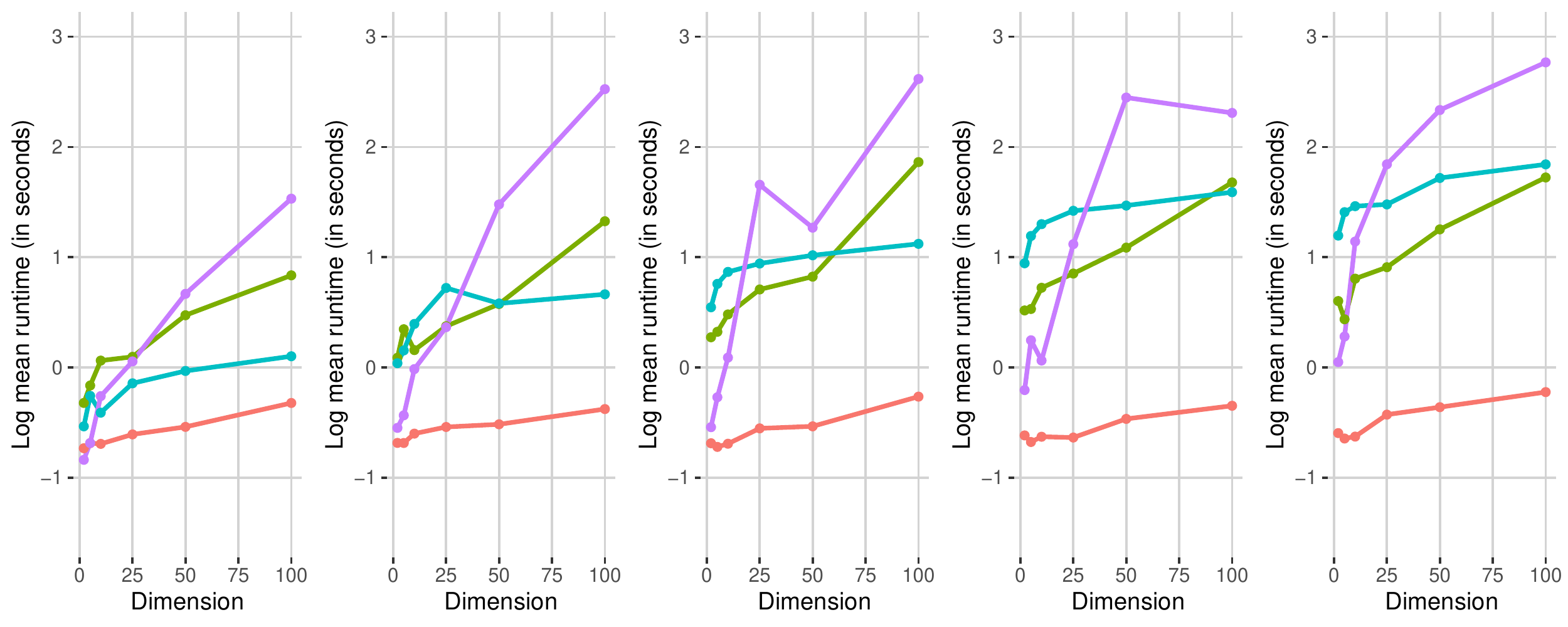}%
}

\caption{Comparison of run times for different algorithms for the ES risk budgeting problem: \tikzcircle[colSolver1, fill=colSolver1]{2pt} Cutting Planes, \tikzcircle[colSolver2, fill=colSolver2]{2pt} IpOpt, \tikzcircle[colSolver3, fill=colSolver3]{2pt} MOSEK, \tikzcircle[colSolver4, fill=colSolver4]{2pt} SCS. From left to right, $N=1\,000, \ 2\,000, \ 3\,000, \ 4\,000, \ 5\,000.$}
\label{fig:cvar_gaussian}
\end{figure}

\begin{figure}[!htp]
\centering
\subfloat[Multivariate Student $t$ model with $\alpha = 0.95$.]{%
  \includegraphics[width = 0.8\textwidth]{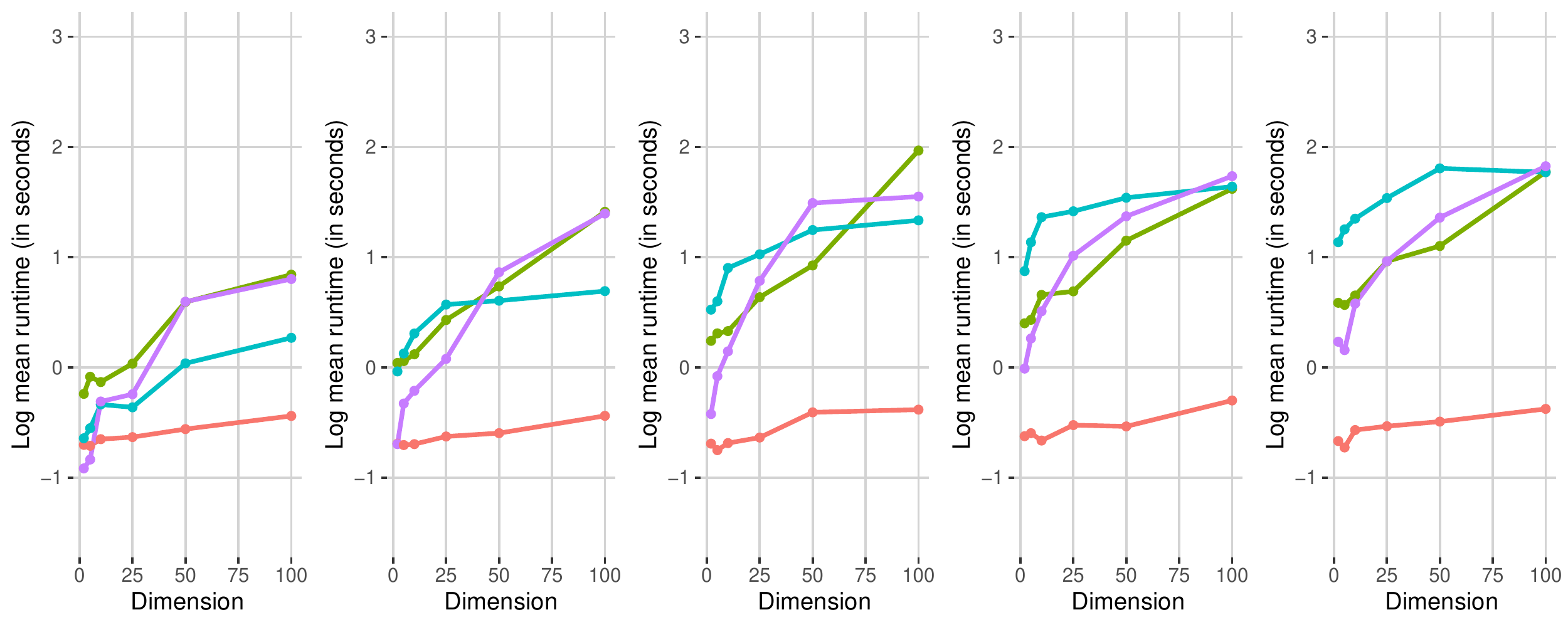}%
}\quad
\subfloat[Multivariate Student $t$ model with $\alpha = 0.99$.]{%
  \includegraphics[width = 0.8\textwidth]{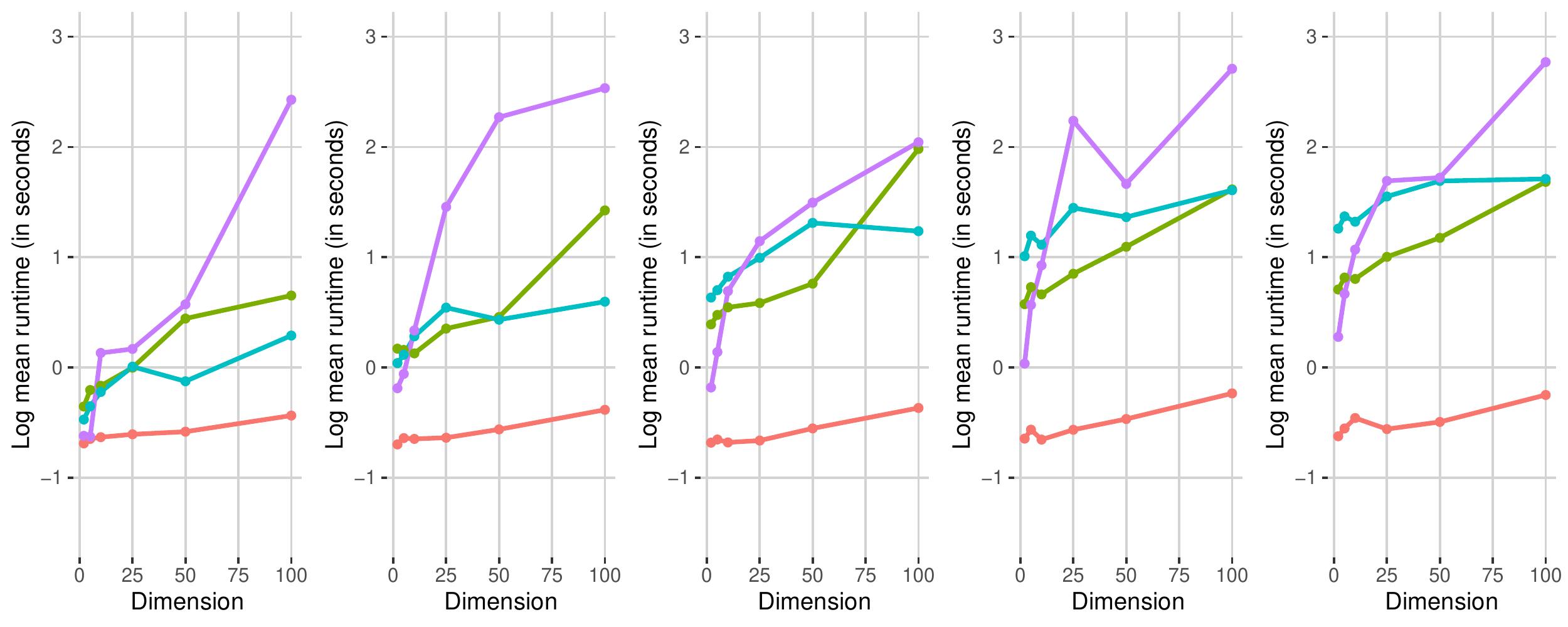}%
}
\caption{Comparison of run times for different algorithms for the ES risk budgeting problem: \tikzcircle[colSolver1, fill=colSolver1]{2pt} Cutting Planes, \tikzcircle[colSolver2, fill=colSolver2]{2pt} IpOpt, \tikzcircle[colSolver3, fill=colSolver3]{2pt} MOSEK, \tikzcircle[colSolver4, fill=colSolver4]{2pt} SCS. From left to right, $N=1\,000, \ 2\,000, \ 3\,000, \ 4\,000, \ 5\,000.$}
\label{fig:cvar_t}
\end{figure}

\begin{figure}[ht]%
	\centering
	\includegraphics[width = 0.8\textwidth]{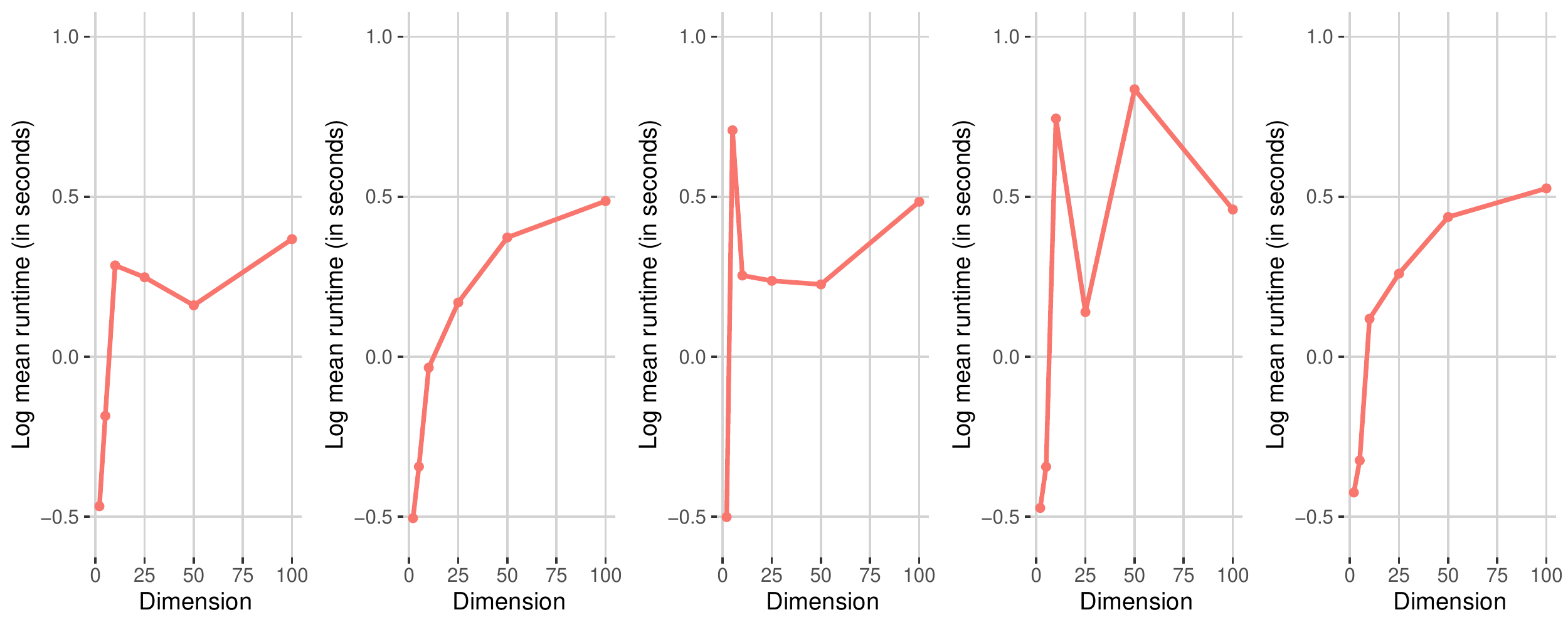}%
	\caption{Run times for the Cutting Planes algorithm for the EVaR risk budgeting problem.
    From left to right, $N=1\,000, \ 2\,000, \ 3\,000, \ 4\,000, \ 5\,000.$ All plots are based on the multivariate Gaussian model.}%
	\label{fig:entropicVaR_dim_time_TRUE}%
\end{figure}

\begin{figure}[ht]%
	\centering
	\includegraphics[width = 0.8\textwidth]{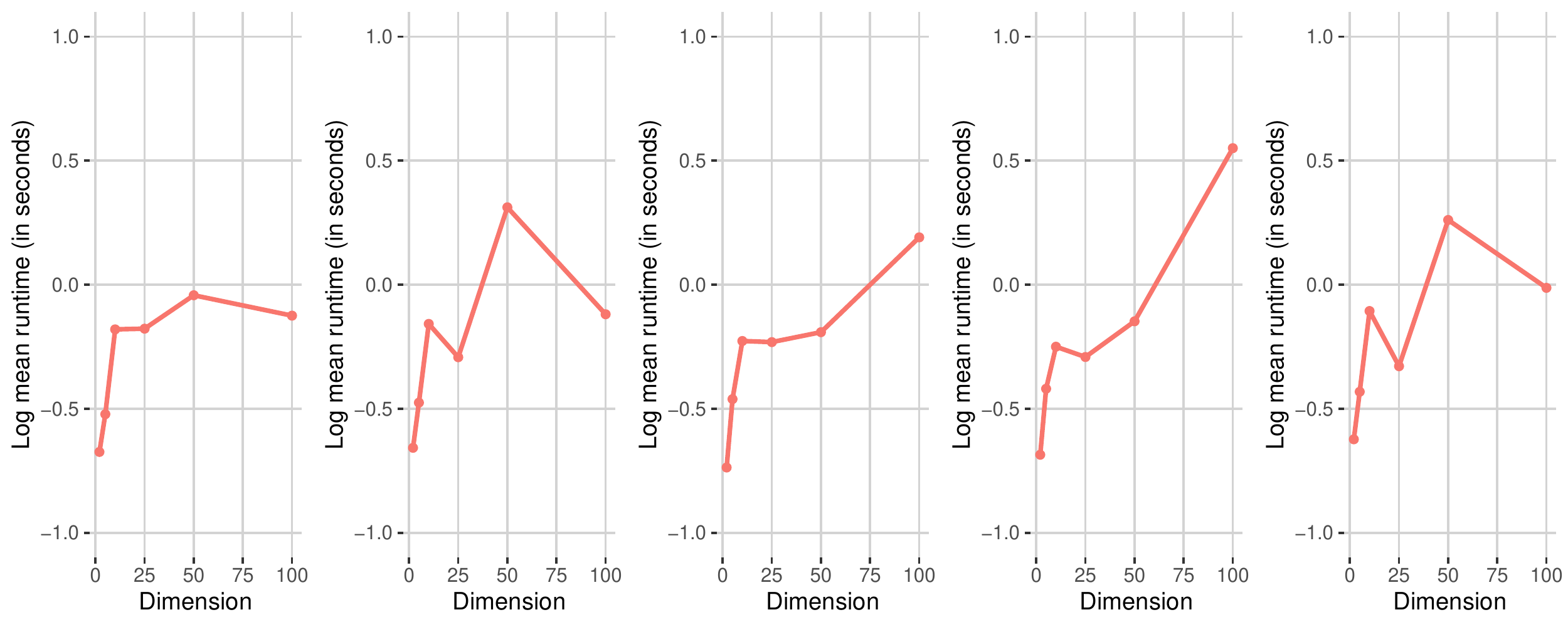}%
	\caption{Run times for the Cutting Planes algorithm for the distortion risk budgeting problem. From left to right, $N=1\,000, \ 2\,000, \ 3\,000, \ 4\,000, \ 5\,000.$ All plots are based on the multivariate Gaussian model.}%
	\label{fig:distortion_dim_time_FALSE}%
\end{figure}

\clearpage
\section{Summary statistics} \label{sec:appendix_summary_statistics}
Table \ref{tbl:portfolio_table_stats} reports the annualised volatility, annualised return, and annual Sharpe ratio, as well as the $\text{VaR}_{0.05}$, $\text{ES}_{0.05}$, and the maximum drawdown of all considered portfolios. We observe that the statistics for the risk parity portfolios with different ES security level do not change drastically.
\vspace{-0.2cm}
\begin{table}[h!]
\footnotesize
\centering
\begin{tabular}{lrrrrrr}
  \toprule\toprule
 		& Volatility & Return & Sharpe & $\text{VaR}_{0.05}$ & $\text{ES}_{0.05}$ &  Max drawdown \\ 
  \midrule
   \sdrp{}  & 18.44\% & 9.37\% & 0.508 & -1.67\% & -2.91\% & -59.89\% \\ 
  \msr{}   & 23.54\% & 12.88\% & 0.547 & -2.05\% & -3.67\% & -70.75\% \\ 
  \gmv{}   & 17.02\% & 9.71\% & 0.571 & -1.56\% & -2.69\% & -58.36\% \\ 
  \mmv{}   & 26.87\% & 14.79\% & 0.551 & -2.56\% & -4.21\% & -66.58\% \\ 
  \ew{}    & 18.89\% & 8.80\% & 0.466 & -1.69\% & -2.97\% & -59.48\% \\ 
  \midrule
  \grp80 & 18.49\% & 9.40\% & 0.508 & -1.70\% & -2.91\% & -60.74\% \\ 
  \grp85 & 18.49\% & 9.38\% & 0.507 & -1.70\% & -2.91\% & -60.62\% \\ 
  \grp90 & 18.50\% & 9.49\% & 0.513 & -1.71\% & -2.91\% & -60.38\% \\ 
  \grp95 & 18.51\% & 9.58\% & 0.517 & -1.69\% & -2.91\% & -60.36\% \\ 
  \grp96 & 18.53\% & 9.45\% & 0.510 & -1.69\% & -2.92\% & -60.32\% \\ 
  \grp97 & 18.52\% & 9.33\% & 0.504 & -1.70\% & -2.92\% & -60.24\% \\ 
  \grp98 & 18.53\% & 9.23\% & 0.498 & -1.69\% & -2.92\% & -60.23\% \\ 
  \grp99 & 18.56\% & 9.43\% & 0.508 & -1.69\% & -2.92\% & -59.99\% \\ 
  \midrule
  \rpa80 & 18.89\% & 8.89\% & 0.471 & -1.69\% & -2.99\% & -60.26\% \\ 
  \rpa85 & 18.82\% & 8.90\% & 0.473 & -1.68\% & -2.98\% & -60.40\% \\ 
  \rpa90 & 18.79\% & 8.93\% & 0.475 & -1.68\% & -2.97\% & -60.34\% \\ 
  \rpa95 & 18.88\% & 8.93\% & 0.473 & -1.69\% & -2.99\% & -60.22\% \\ 
  \rpa96 & 18.96\% & 9.04\% & 0.477 & -1.70\% & -3.00\% & -60.19\% \\ 
  \rpa97 & 19.02\% & 9.03\% & 0.475 & -1.72\% & -3.01\% & -60.28\% \\ 
  \rpa98 & 19.13\% & 8.98\% & 0.469 & -1.72\% & -3.02\% & -59.89\% \\ 
  \rpa99 & 19.33\% & 8.87\% & 0.459 & -1.70\% & -3.06\% & -59.83\% \\ 
  \midrule
  \rpb80 & 18.89\% & 8.89\% & 0.471 & -1.68\% & -2.99\% & -60.34\% \\ 
  \rpb85 & 18.83\% & 8.85\% & 0.470 & -1.67\% & -2.98\% & -60.27\% \\ 
  \rpb90 & 18.76\% & 8.90\% & 0.474 & -1.66\% & -2.97\% & -60.37\% \\ 
  \rpb95 & 18.71\% & 8.78\% & 0.469 & -1.66\% & -2.96\% & -60.25\% \\ 
  \rpb96 & 18.75\% & 8.77\% & 0.468 & -1.66\% & -2.97\% & -60.28\% \\ 
  \rpb97 & 18.75\% & 8.93\% & 0.476 & -1.68\% & -2.96\% & -60.19\% \\ 
  \rpb98 & 18.77\% & 9.00\% & 0.479 & -1.67\% & -2.97\% & -60.21\% \\ 
  \rpb99 & 18.76\% & 8.90\% & 0.474 & -1.66\% & -2.96\% & -59.69\% \\  \bottomrule\bottomrule
\end{tabular}
\caption{Statistics for all portfolios. The volatility, return and Sharpe ratio are annualised.}
\label{tbl:portfolio_table_stats}
\end{table}

\end{document}